\documentclass{article}

\usepackage{fullpage}
\usepackage{authblk}
\usepackage{algorithm}
\usepackage[utf8]{inputenc}
\usepackage{amsmath, amsfonts, amssymb, amsthm}
\usepackage{enumerate}
\usepackage{graphicx}
\usepackage{url}
\usepackage{csquotes}
\usepackage[hidelinks]{hyperref}
\usepackage{enumitem}

\setlist[itemize]{noitemsep, nolistsep,topsep=1pt}
\setlist[enumerate]{noitemsep, nolistsep,topsep=1pt}

\usepackage[labelfont=bf, font=small]{caption}
\usepackage{booktabs}
\usepackage{rotating}
\usepackage{soul}
\usepackage{xspace}
\makeatletter
\renewcommand{\fnum@figure}{Fig. \thefigure}
\makeatother
\usepackage{thm-restate}

\usepackage{todonotes}
\usepackage{subcaption}

\usepackage{tikz}
\usetikzlibrary{arrows.meta, calc, decorations.markings}
\usepackage{xcolor}
\usepackage{mathtools}

\usepackage[normalem]{ulem}
\definecolor{changecolor}{RGB}{192,64,0}
\definecolor{changedelcolor}{RGB}{140,140,0}

\newtheorem{theorem}{Theorem}[section]
\newtheorem{lemma}[theorem]{Lemma}
\newtheorem{corollary}[theorem]{Corollary}
\newtheorem{proposition}[theorem]{Proposition}

\newtheorem{observation}{Observation}[section]

\theoremstyle{remark}

\theoremstyle{definition}
\newtheorem{definition}[theorem]{Definition}

\hypersetup{
        colorlinks=true,
        linkcolor=blue,
        filecolor=magenta,      
        urlcolor=blue,
        pdftitle={Overleaf Example},
        pdfpagemode=FullScreen,
        citecolor=blue
}

\usepackage[capitalise, nameinlink]{cleveref}

\crefname{section}{Section}{Sections}
\crefname{figure}{Figure}{Figures}
\crefname{table}{Table}{Tables}
\crefname{lemma}{Lemma}{Lemmas}
\crefname{theorem}{Theorem}{Theorems}
\crefname{corollary}{Corollary}{Corollaries}
\crefname{observation}{Observation}{Observations}
\crefname{proposition}{Proposition}{Propositions}
\crefname{conjecture}{Conjecture}{Conjectures}
\crefname{example}{Example}{Examples}
\crefname{definition}{Definition}{Definitions}
\crefname{remark}{Remark}{Remarks}

\newcommand{\NN}{\mathfrak{N}}
\newcommand{\BB}{\mathfrak{B}}

\newcommand{\PP}{\mathcal{P}}
\newcommand{\CC}{\mathcal{C}}
\newcommand{\canonical}[1]{#1^{c}}
\newcommand{\QQ}{\mathcal{Q}}

\newcommand{\rhotilde}{\widetilde{\rho}}
\newcommand{\dtilde}{\widetilde{d}}
\newcommand{\T}{\mathcal{T}}
\newcommand{\TT}{\widetilde{\mathcal{T}}}
\newcommand{\m}{\mu}
\newcommand{\eqover}[1]{\stackrel{\mathclap{\normalfont\scriptsize\mbox{#1}}}{=}}
\newcommand{\Split}{\mathrm{Split}}
\newcommand{\SSS}{\mathcal{S}}
\newcommand{\supp}{\mathrm{supp}}

\usepackage[backend=biber,
            giveninits=true,
            maxbibnames=12,
            style=alphabetic,
            sortcites]{biblatex}
\title{Distinguishing Phylogenetic Level-2 Networks with Quartets and Inter-Taxon Quartet Distances}

\author[1,*]{Niels~Holtgrefe}
\author[2]{Elizabeth~S.~Allman}
\author[3]{Hector~Ba\~{n}os}
\author[1]{Leo~van~Iersel}
\author[4]{Vincent~Moulton}
\author[2]{John~A.~Rhodes}
\author[5]{Kristina~Wicke}

\affil[1]{Delft Institute of Applied Mathematics, Delft University of Technology, Delft, The Netherlands}
\affil[2]{Department of Mathematics and Statistics, University of Alaska Fairbanks, Fairbanks, AK, USA}
\affil[3]{Department of Mathematics, California State University San Bernardino, San Bernardino, CA, USA}
\affil[4]{School of Computing Sciences, University of East Anglia, Norwich, United Kingdom}
\affil[5]{Department of Mathematical Sciences, New Jersey Institute of Technology, Newark, NJ, USA}

\date{}

\begin{document}

\maketitle

\begingroup
\renewcommand\thefootnote{}\footnotetext{*Corresponding author\\ \emph{Email addresses:} 
\href{mailto:n.a.l.holtgrefe@tudelft.nl}{n.a.l.holtgrefe@tudelft.nl},
\href{mailto:e.allman@alaska.edu}{e.allman@alaska.edu}, \href{mailto:hector.banos@csusb.edu}{hector.banos@csusb.edu},
\href{mailto:l.j.j.vaniersel@tudelft.nl}{l.j.j.vaniersel@tudelft.nl},
\href{mailto:v.moulton@uea.ac.uk}{v.moulton@uea.ac.uk},
\href{mailto:j.rhodes@alaska.edu}{j.rhodes@alaska.edu},
\href{mailto:kristina.wicke@njit.edu}{kristina.wicke@njit.edu} 
}
\endgroup

\begin{abstract}
The inference of phylogenetic networks, which model complex evolutionary processes including hybridization and gene flow, remains a central challenge in evolutionary biology. Until now, statistically consistent inference methods have been limited to phylogenetic level-1 networks, which allow no interdependence between reticulate events. In this work, we establish the theoretical foundations for a statistically consistent inference method for a much broader class: semi-directed level-2 networks that are outer-labeled planar and galled. We precisely characterize the features of these networks that are distinguishable from the topologies of their displayed quartet trees. Moreover, we prove that an inter-taxon distance derived from these quartets is circular decomposable, enabling future robust inference of these networks from quartet data, such as concordance factors obtained from gene tree distributions under the Network Multispecies Coalescent model. Our results also have novel identifiability implications across different data types and evolutionary models, applying to any setting in which displayed quartets can be distinguished.
\end{abstract}

\section{Introduction}

As analysis of genomic datasets has advanced, growing evidence has indicated that hybridization and other reticulate events 
play a significant role in evolution (see, e.g.,~\cite{swithers2012role, deBaun2023widespread,sessa2012}). Hence it is 
desirable, and increasingly common, to rely on phylogenetic networks instead of phylogenetic trees to depict evolutionary relationships between species.
We consider semi-directed (phylogenetic LSA)
networks in which the root is suppressed, since the root location cannot be identified under many models of evolution and data types \cite{solislemus2016-snaq,
banos2019identifying,Ane2023-anomalousnetworks,gross2021-markovlvl1,xu2022-pairwisedistiden}, and such networks might be rooted in practice using  
an outgroup \cite{kilman2016-rooting}. Limited by both the lack of theoretical identifiability results for complex networks, as well as computational constraints, 
existing statistically consistent inference methods focus on the class of semi-directed \emph{level-1} networks \cite{solislemus2016-snaq,allman2019nanuq,kong2024-phynest,holtgrefe2025squirrel,allman2024nanuq+}. 
Within this class of networks, reticulate events are assumed to be isolated from each other, 
forcing the network to consist of only tree-like parts and disjoint cycles. Numerous identifiability results have been obtained for this restricted 
class of networks, proving that theoretically (most of) a level-1 network can be distinguished from biological data under several evolutionary models \cite{gross2021-markovlvl1,banos2019identifying,xu2022-pairwisedistiden,allman2022-logdetlvl1,allman2024-numericidentlvl1,englander2025identifiability}.

Several results have shown that certain features of higher-level networks 
(informally, networks with a higher degree of interdependence between reticulate events) can be identified from biological data. 
Specifically, the \emph{tree-of-blobs} --- {a tree depicting only the tree-like branching structure of a network --- 
is identifiable from gene tree probabilities under several variants of the Network Multispecies Coalescent model 
\cite{allman2023tree,rhodes2025identifying}, 
as well as from nucleotide sequences 
under the Jukes-Cantor and Kimura-2-Parameter models \cite{englander2025identifiability}. Similarly, when a \emph{blob} 
(a 2-edge-connected component of the network) is \emph{outer-labeled planar} 
(i.e., it can be drawn in the plane without crossing edges and with the taxon labels on the outside), 
the circular order of its pendant subnetworks is identifiable from the same types of data and under the same models \cite{rhodes2025identifying,englander2025identifiability}.

Two recent breakthroughs have shown the identifiability of the full structure of some higher-level networks. 
On the one hand, \cite{englander2025identifiability} 
showed that under the Jukes-Cantor model, semi-directed binary, strongly tree-child level-2 networks can be identified from nucleotide sequences if the networks contain no 3-cycles within blobs. On the other hand, \cite{AABR2025galledTC} showed identifiability of certain semi-directed binary, galled, strongly tree-child level-$k$ networks from gene tree probabilities under variations of the Network Multispecies Coalescent model.  
Both these results demonstrate that theoretically sound inference methods for networks more general than level-1 are possible, though neither suggest a particular computationally tractable inference procedure.

\medskip

In this paper, we take a first step in this direction by providing the theoretical foundations for a statistically consistent inference method for a class of level-2 networks. In particular, our main result (\cref{thm:identifiability}) revolves around a canonical form that characterizes exactly when semi-directed level-2 networks that are both galled and outer-labeled planar (see \cref{fig:humans_example} 
for an example of such a network) can and cannot be distinguished by means of displayed quartets. Hence, this canonical form is the most refined network that can theoretically be inferred when using displayed quartets alone, and thus also illuminates the theoretical limits of level-2 network inference methods relying only on quartets. Since our result does not rely on any specific model, immediate identifiability results follow
from biological data generated under models of evolution that meet our criteria. In particular, our result 
shows that the canonical form can be identified under any model in which the topology of displayed quartets can be identified. This includes all aforementioned models \cite{englander2025identifiability,rhodes2025identifying}.

\begin{figure}[htb]
    \centering
    \includegraphics[width=0.25\textwidth]{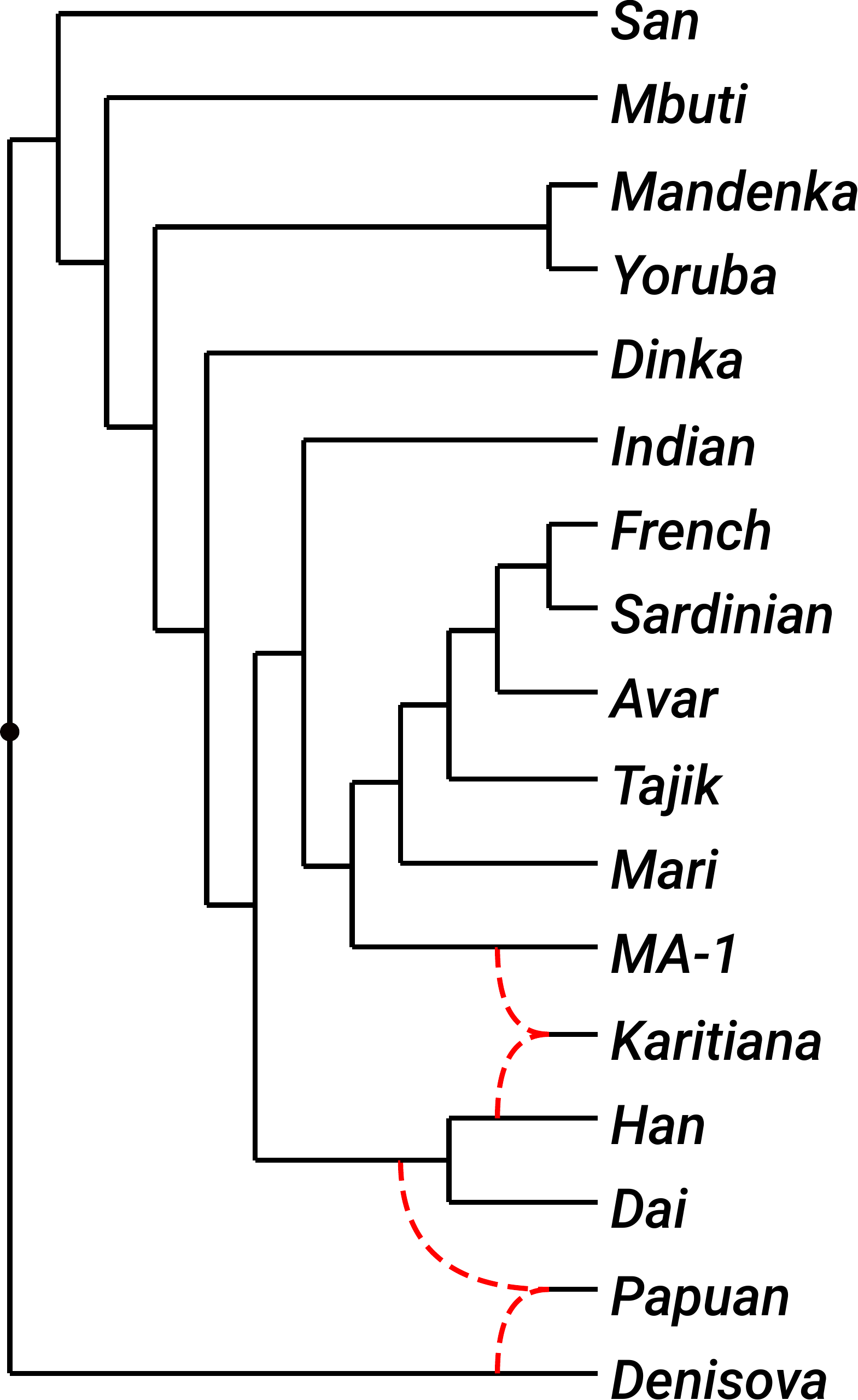}
    \qquad\qquad
    \includegraphics[width=0.45\textwidth]{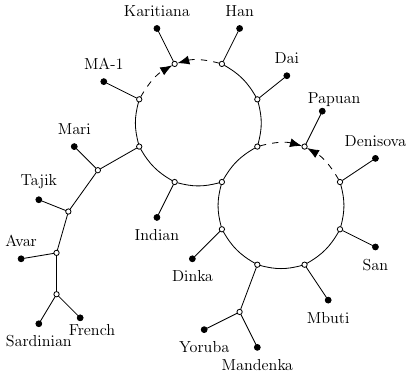}
    \caption{\emph{Left:} A rooted phylogenetic level-2 network on 17 human populations \cite{raghavan2014upper}, where the directions go from left to right and the root is located at the black dot. The network was constructed from 16 complete genomes from modern human worldwide and \emph{MA-1}, a 24,000-year-old anatomically modern human from the Mal'ta site in south-central Siberia. \emph{Right:} The semi-directed phylogenetic network obtained from the rooted network by suppressing its root and only retaining directions of the dashed hybrid edges. The network is level-2, outer-labeled planar, and galled.}
    \label{fig:humans_example}
\end{figure}

Our identifiability proof is constructive in nature and 
at its heart lies a generalization of the inter-taxon quartet distances 
associated to semi-directed level-1 networks\footnote{See also \cite{holtgrefe2025squirrel} for a similar distance.}
from \cite{allman2019nanuq,allman2024nanuq+}. These distances have the advantage that they 
can be computed directly from data such as \emph{concordance factors}: summaries of gene tree distributions under the Network Multispecies Coalescent model. 
We prove that these distances are \emph{circular decomposable} for so-called \emph{bloblets}: networks with a single internal blob (see \cref{thm:NANUQ_circular}). Essentially, these distances can be decomposed into sums of simpler metrics in a way that respects
the circular ordering of the taxa induced by the planar embedding of the network.  
Interestingly, this approach directly leads to 
an inference algorithm 
for computing outer-labeled planar, galled, level-2 networks in the spirit of the existing level-1 inference tools NANUQ and NANUQ$^+$ \cite{allman2019nanuq,allman2024nanuq+}.
An outline of this algorithm is sketched in the discussion; specifics 
and its implementation will follow in future work.

This paper is structured as follows. \cref{sec:preliminaries} covers preliminaries on phylogenetic networks, split systems, circular decomposable metrics, and the quartet metric from \cite{rhodes2019topological}. In \cref{sec:displ_quartet_metric} we prove some auxiliary results used in \cref{sec:NANUQ_metric} to prove circular decomposability of the inter-taxon quartet distances under consideration. \cref{sec:identifiability} contains the main result about our canonical form and the biological identifiability results that follow from it. We end with a discussion in \cref{sec:discussion}.

\section{Preliminaries}\label{sec:preliminaries}

\subsection{Phylogenetic networks}

We use standard terminology for phylogenetic networks as in \cite{steel2016-book}.

\begin{definition}[Rooted phylogenetic network]
A \emph{(binary) rooted (phylogenetic) network} $N^+$ on a set of at least two taxa $X$ is a rooted directed acyclic graph such that (i) the (unique) root has in-degree zero and out-degree two; (ii) its leaves are of in-degree one and out-degree zero and they are bijectively labeled by elements of $X$; (iii) all other nodes either have in-degree one and out-degree two (known as \emph{tree nodes}), or in-degree two and out-degree one (known as \emph{hybrid nodes}).
\end{definition}

A rooted network admits a natural partial order of its nodes. The \emph{least stable ancestor (LSA)} of a set of leaves $Y$ of a rooted network is the lowest node through which all directed paths from the root to any leaf in $Y$ must pass. Throughout this work, we assume that for all rooted networks on $X$, the root is the LSA of $X$ (also known as \emph{LSA networks}). The two edges directed towards a hybrid node are called \emph{hybrid edges}. We call a leaf in $X$ a \emph{hybrid (leaf)} if its parent is a hybrid node.

\begin{definition}[Semi-directed phylogenetic network]\label{def:semi_directed}
A \emph{(binary) semi-directed (phylogenetic) network} $N$ on $X$ is a partially directed graph that 
can be obtained from a binary rooted phylogenetic LSA network $N^+$ by undirecting all non-hybrid edges and suppressing the former root.
\end{definition}

In the definition above, the rooted network  $N^+$ is
called a \emph{rooted partner} of $N$ (see also \cref{fig:example_network}(a) and~(b)). We define hybrid nodes and edges in semi-directed networks analogously to those for rooted networks, and note 
 that a semi-directed network may have parallel edges.
A semi-directed network without hybrid nodes is a \emph{(binary unrooted) phylogenetic tree}.

Since phylogenetic networks, whether rooted or semi-directed, never have cycles in the (semi-)directed sense, we use the term \emph{cycle} to mean nodes forming a cycle in the fully undirected network. For instance, the level-2 networks in \Cref{fig:humans_example} each have three cycles, with the semi-directed network on the right showing two 8-cycles, while the third is a 14-cycle that includes all nodes within the large blob.

\begin{figure}[ht]
    \centering
    \includegraphics[width=0.82\textwidth]{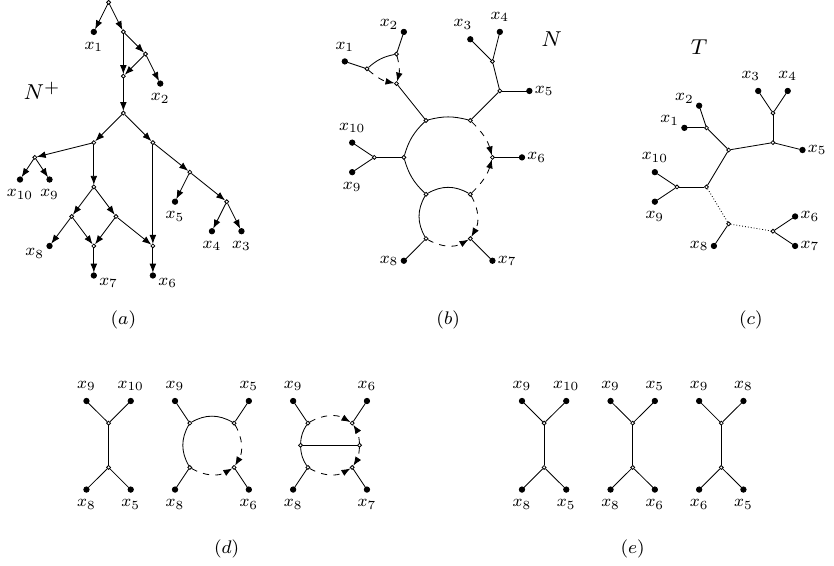}
    \caption{$(a)$: A rooted phylogenetic network~$N^+$ on $X = \{x_1, \ldots, x_{10} \}$.
    $(b)$: The semi-directed phylogenetic network~$N$ on $X = \{x_1, \ldots, x_{10} \}$ obtained from $N^+$. The network $N$ is outer-labeled planar, strictly level-2 and galled.
    $(c)$: A phylogenetic tree~$T$ on $X = \{x_1, \ldots, x_{10} \}$ that is displayed by $N$ with multiplicity~2. The tree that can be obtained from $T$ by contracting the two dotted edges is the tree-of-blobs of~$N$.
    $(d)$: Three quarnets induced by $N$.
    $(e)$: Three displayed quartets of $N$.
    }
    \label{fig:example_network}
\end{figure}

A \emph{blob} of a rooted or semi-directed network is a maximal weakly connected subgraph without any cut edges.
An \emph{articulation node} of a blob is a node in the blob that is incident to
a cut edge of the network. A blob with $m$ articulation nodes is called an
\emph{$m$-blob}. By \emph{contracting} a blob, we mean replacing the blob by a single node, and suppressing it in case this node has degree-2. The \emph{tree-of-blobs} of a semi-directed network is the phylogenetic tree 
that can be obtained from it by contracting every blob (see  \cref{fig:example_network}(c)). 
In the following definition, we list some additional properties that semi-directed networks may have.

\begin{definition}
Let $N$ be a semi-directed network. Then, we call $N$
\begin{enumerate}[label={(\roman*)}, noitemsep]
\item \emph{level-$k$} for some non-negative integer $k$ if there exist at most $k$ hybrid nodes in each blob of the network;
\item \emph{strictly level-$k$} for some positive integer $k$ if the network is level-$k$ but the network is not level-$(k-1)$;
\item a \emph{bloblet (network)}\footnote{Such networks have also been called \emph{simple}.} if the network has a single non-leaf blob;
\item \emph{$k$-cycle-free} for some non-negative integer $k$ if the network contains no cycles of length~$k$;
\item \emph{outer-labeled planar} if the network has a planar embedding with the leaves on the unbounded face;
\item \emph{galled} if for every hybrid node $h$ and every pair of partner hybrid edges $e =(v,h)$ and $e'=(v',h)$,
there exists a cycle in the network (disregarding edge directions)
that contains $e$ and $e'$ and no other hybrid edges of the network.
\end{enumerate}
\end{definition}

Note that a galled network can equivalently be defined as one with no hybrid node \emph{ancestral} to another hybrid node in the same blob (i.e., there is no pair of hybrid nodes in the same blob having a path between them in which all directed edges have the same orientation), 
or in the binary setting of this article, as one where each hybrid node is an articulation node.
We also extend the notions of \emph{level} and \emph{strict level} to individual blobs by counting the number of hybrid nodes in that particular blob. 

The specific classes of networks we consider are given in the next definition. \cref{fig:example_network}(a) and (b) depict 
a semi-directed network from the class $\NN'_2$ and one of its rooted partners.
\begin{definition}
     The (binary) semi-directed phylogenetic networks that are outer-labeled planar, galled, and level-$k$ (for some $k\geq 0$) form the class $\NN_k$. 
     The strictly level-$k$ networks in  $\NN_k$ form the subclass  $\NN'_k$.
    The bloblets in $\NN_k$ and $\NN'_k$, respectively, form the subclasses $\BB_k$ and $\BB'_k$.  
\end{definition}

An \emph{up-down path} between two labeled leaves $x_1$ and $x_2$ of a semi-directed network is a path of $k$ edges where the first $\ell$ edges are directed towards $x_1$ and the last $k - \ell$ edges are directed towards $x_2$, where  undirected edges are considered bidirected.
\begin{definition}[Subnetwork]\label{def:subnetwork}
Given a semi-directed network $N$ on $X$ and some  $Y \subseteq X$ with $| Y | \geq 2$, the \emph{subnetwork of $N$ induced by $Y$} is the semi-directed network $N|_Y$ obtained from $N$ by taking the union of all up-down paths between leaves in $Y$, followed by exhaustively suppressing all degree-2 nodes.    
\end{definition}
If $|Y| = 4$ with $Y = \{x,y,z,w\}$, we may write $N|_{xyzw}$ to mean $N|_{\{x,y,z,w\}}$. Such a network is called a \emph{quarnet}, or a \emph{quartet} in case $N$ is an unrooted phylogenetic tree (see \cref{fig:example_network}(d)).

Given a semi-directed network $N$ on $X$ and a phylogenetic tree $T$ on $X$, we say that $T$ is \emph{displayed} by $N$ if it can be obtained from $N$ by deleting exactly one hybrid edge per hybrid node, then 
recursively deleting all leaves not in $X$ and exhaustively suppressing degree-2 nodes. We use $\T (N)$ to denote the set of phylogenetic trees displayed by $N$. The \emph{multiplicity} $\mu (T , N)$ of $T$ in $\T (N)$ is the number of distinct choices of hybrid edges in $N$ that when deleted lead to the displayed tree $T$. If $T \not \in \T (N)$, then we set $\mu (T, N) = 0$. We write $\mu (N) = \sum_{T \in \T (N)} \mu (T, N)$, and note that if $N$ has $r$ hybrid nodes, then $|\T(N)| \leq \m (N) = 2^r$. 
Lastly, we refer to a quartet tree that is displayed by a quarnet of $N$ as a \emph{displayed quartet of $N$}. See \cref{fig:example_network}(c) and~(e).

\subsection{Split systems and circular metrics}

This section introduces several known concepts related to splits and metrics, see e.g. \cite{bryant2007consistency,allman2019nanuq}. Recall that a \emph{pseudo metric} $d : X^2 \rightarrow \mathbb{R}_{\geq 0}$ on a finite set of elements $X$ is a symmetric non-negative function that (i) satisfies the triangle inequality; and (ii) has the property that $d(x,x)=0$ for all $x \in X$. In particular, for $x,y \in X$, $d(x,y)=0$ need not imply $x=y$.
All results in this work revolve around pseudo metrics, which we will simply use the word \emph{metrics} for convenience unless we want to emphasize the fact that we are considering pseudo metrics.

A \emph{split} $A|B = B|A$ of a finite set $X$ with at least two elements is a bipartition of $X$ with $A, B \subseteq X$. A split is \emph{empty} if $A$ or $B$ has size 0, and it is \emph{trivial} if $A$ or $B$ has size 1. We say that a cut edge of a semi-directed network $N$ \emph{induces} the split $A|B$ of its leaf set $X$ if the removal of the cut edge disconnects the leaves labeled by elements in $A$ from those labeled by elements in~$B$. We denote the set of splits induced by the cut edges of 
a semi-directed network~$N$ as $\Split (N)$. Similarly, we denote the set of splits induced by the displayed trees of $N$ by $\Split (\T (N))$. To distinguish between the two, we call the former the \emph{induced splits of~$N$} and the latter the \emph{displayed splits of~$N$}. Clearly, $\Split(N) \subseteq \Split(\T(N))$.

The set of all nonempty splits of a finite set $X$ is denoted by $\Split (X)$, and we call a non-empty set $\SSS \subseteq \Split (X)$ a \emph{split system} on $X$. A map $\omega : \Split (X) \rightarrow \mathbb{R}_{\geq 0}$ is a \emph{weighted split system}. We let the \emph{support} of the weighted split system, denoted by $\supp (\omega)$, be the split system $\SSS \subseteq \Split (X)$ containing each split that has a positive weight $\omega$. Note that any weighted split system $\omega: \Split (X) \rightarrow \mathbb{R}_{\geq 0}$ induces a metric $d_\omega$ on $X$ as follows. Given some split $ A|B \in \Split (X)$, we define the \emph{split metric} $\delta_{A|B}$ as
\begin{equation}
    \delta_{A|B} (x,y) =
\begin{cases} 
0, & \text{if } x,y \in A \text{ or } x,y \in B; \\ 
1, & \text{otherwise.}
\end{cases}
\end{equation}
Then, the metric $d_\omega$ \emph{induced by $\omega$} is
\begin{equation}
    d_\omega (x, y) = \sum_{S \in \Split (X)} \omega(S) \cdot \delta_S (x,y).
\end{equation}
It is easy to see that the triangle inequality holds for $d_\omega$, so that it is indeed a (pseudo) metric.

Before the next definition, we define a \emph{circular order} $\CC$ of a finite set $X = \{x_1, \ldots, x_n\}$ as an ordering of its elements up to reversal and cyclic permutations. We often denote such a circular order as $\CC = (x_0, x_1, \ldots, x_n = x_0)$ or simply $\CC = (x_1, \ldots, x_n)$.

\begin{definition}[Circular split system]\label{def:circular_splitsystem}
A split system $\SSS \subseteq \Split (X)$ is \emph{circular} if there exists a circular order
$\CC = (x_0, x_1, \ldots, x_n = x_0)$ of the finite set $X = \{x_1, \ldots, x_n \}$ such that each split in $\SSS$ has the form $S_{ij} (\CC) = U_{ij} | V_{ij}$ for some $i \neq j$ with $U_{ij} = \{x_{i+1}, \ldots, x_j \}$ and $V_{ij} = \{x_{j+1}, \ldots, x_i \}$. We say that such a circular order is \emph{congruent} with $\SSS$.
\end{definition}

We sometimes write $S_{ij}$ instead of $S_{ij} (\CC)$ if $\CC$ is clear from the context. The circular split system consisting of all splits $S_{ij} (\CC)$
for a circular order $\CC$ is denoted by  $\SSS(\CC)$.
For example, \cref{fig:example_circular_splits} gives a visualization of a circular split system $\SSS \subset \SSS\big((x_1, \dots, x_n)\big)$. Note that the split system in the figure is also congruent with other circular orders not depicted. We say that a weighted split system is \emph{circular} if its support is a circular split system. 

\begin{figure}[h!]
    \centering
    \includegraphics[width=0.55\textwidth]{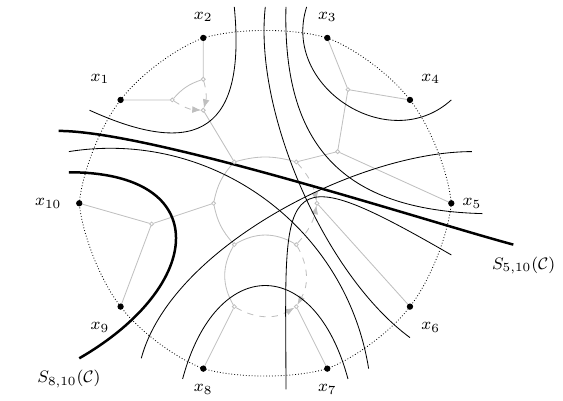}
    \caption{A visualization of the circular split system $\SSS \subset \Split (\T (N))$, containing all non-trivial displayed splits (depicted by the black lines) of the outer-labeled planar, semi-directed network~$N$ on $X = \{x_1, \ldots, x_{10}\}$ (depicted in gray) from \cref{fig:example_network}(b). $\SSS  \subset \SSS(\CC)$ is congruent with the circular order $\CC = (x_1, \ldots x_{10})$ (depicted by the dotted lines), one of the induced circular orders of~$N$. Two splits are highlighted in thick black and are labeled using the notation from \cref{def:circular_splitsystem}.}
    \label{fig:example_circular_splits}
\end{figure}

\begin{definition}[Circular decomposable metric]\label{def:circular_decomposable}
    A metric $d: X^2 \rightarrow \mathbb{R}_{\geq 0}$ is \emph{circular decomposable}, if there exists a circular weighted split system 
    $\omega: \Split(X) \rightarrow \mathbb{R}_{\geq 0}$ such that 
    $d = d_\omega$.
\end{definition}
In other words, a metric $d$ is circular decomposable if it can be decomposed as a sum of weighted split metrics, all corresponding to one circular split system. As shown in \cite{bandelt1992canonical} (see also \cite{bryant2007consistency}), if $d$ is circular decomposable, this decomposition is the unique way to decompose $d$ into a sum of weighted (weakly compatible) split metrics, i.e., the circular weighted split system $\omega$ is unique. Hence, we may also define the \emph{support} $\supp (d_\omega)$ of a circular decomposable metric $d_\omega$ as the unique support of the corresponding circular weighted split system $\omega$.

To prove that a metric is circular decomposable we use the following proposition, which follows from a lemma of Chepoi and Fichet \cite{chepoi1998note}. To avoid confusion, when referring to the split weights defined in the following proposition, we may sometimes write $\alpha_{ij} (d)$ or $\alpha_{ij} (d, \CC)$ instead of simply $\alpha_{ij}$.

\begin{proposition}\label{prop:split_weights}
    Let~$X$ be a finite set of elements, $d: X^2 \rightarrow \mathbb{R}_{\geq 0}$ a pseudo metric, $\CC = (x_0, x_1, \ldots x_n = x_0)$ a circular order of $X$ and $\SSS \subseteq \SSS (\CC)$. For all $S_{ij} \in \SSS (\CC)$, let $\alpha_{ij}$ be the \emph{split weight} of $S_{ij}$ associated with $d$ and $\CC$, which is defined as
    $$\alpha_{ij} = d (x_i,x_j) + d(x_{i+1}, x_{j+1}) - d(x_i, x_{j+1}) - d(x_{i+1}, x_j).$$ Then, $d$ is circular decomposable with support $\SSS$ if and only if $\alpha_{ij} \geq 0$ for all $S_{ij} \in \SSS (\CC)$ and $\SSS = \{S_{ij} \in \SSS (\CC): \alpha_{ij} > 0 \}$.
\end{proposition}
\begin{proof}
    Chepoi and Fichet \cite{chepoi1998note} proved that any symmetric function $d: X^2 \rightarrow \mathbb{R}$ with a zero diagonal can be written as 
\begin{equation}\label{eq:decomp}
        d(x,y) = \frac{1}{2} \sum_{S_{ij} \in \SSS (\CC)} \alpha_{ij} \cdot  \delta_{S_{ij}}(x,y) 
    \end{equation}
    for all $x,y \in X$ and for any circular order~$\CC$ of~$X$. The backward direction then follows. For the forward direction, note that by the main result in \cite{chepoi1998note} (see also \cite{christopher1996structure} for an alternative proof), circular decomposable metrics congruent with a circular ordering~$\CC$ are equivalent to \emph{Kalmanson metrics} \cite{kalmanson1975edgeconvex} compatible with~$\CC$, that is, metrics for which $\alpha_{ij} \geq 0$ for all $S_{ij} \in \SSS (\CC)$. So, if $d$ is circular decomposable with support~$\SSS$, all $\alpha_{ij}$ are non-negative. Therefore, \cref{eq:decomp} gives the circular decomposition of $d$. Since this decomposition is unique \cite{bandelt1992canonical}, it also follows that $\SSS = \{S_{ij} \in \SSS (\CC): \alpha_{ij} > 0 \}$.
\end{proof}

We say that a semi-directed network~$N$ on~$X$ \emph{induces} a circular order~$\CC$ of~$X$ if there exists an outer-labeled planar representation of $N$
that induces the circular order $\CC$ of the leaves in $X$ (see \cref{fig:example_circular_splits}). As shown in \cite{rhodes2025identifying}, outer-labeled planar semi-directed networks induce at least one circular order of their leaf sets. Moreover, if the outer-labeled planar network is a bloblet, this circular order is unique. See also \cite{moulton2022planar} for results related to the planarity of rooted networks. The following result straightforwardly connects the induced circular orders of an outer-labeled planar network to the displayed splits of the network.

\begin{proposition}\label{prop:splits_are_circular}
    Let $N$ be an outer-labeled planar, semi-directed network on $X$. Then, $\Split (\T (N))$ is a circular split system congruent with any circular order of $X$ induced by $N$.
\end{proposition}
\begin{proof}
Let $\CC$ be an induced circular order of $N$. Clearly, every tree $T \in \T(N)$ is also outer-labeled planar and induces the order $\CC$. Then, $\Split (T)$ is a circular split system congruent with $\CC$ for all $T \in \T(N)$ (see, e.g., \cite{Gambette2012}). 
Hence, $\Split (\T (N)) = \bigcup_{T \in \T(N)} \Split (T)$ is a circular split system congruent with $\CC$.
\end{proof}

\subsection{Quartet metric of a phylogenetic tree}\label{subsec:quartet_metric}
We revisit the metric from \cite{rhodes2019topological}, which was defined on the leaves of a phylogenetic tree based on the quartets induced by that tree.

Let $T$ be a phylogenetic tree on four leaves $\{x,y,z,w\}$, i.e.~$T$ is a \emph{quartet tree}. Then, for any pair of leaves $\{x,y\}$, we let
\begin{equation}\label{eq:rho_tree}
\rho_{xy}(T) = \begin{cases}
0 &\text{if } x \text{ and } y \text{ form a cherry},\\
1 &\text{otherwise}.\\
\end{cases}
\end{equation}
Note that $\rho_{xy} (T) = \delta_{S} (x,y)$, where $S$ is the unique non-trivial split induced by $T$.

We can use this concept to define a metric $d_T$ for any phylogenetic tree $T$ with at least four leaves. Recall that $T|_{xyzw}$ denotes the quartet tree on leaf set $\{x,y,z,w\}$ induced by the tree $T$.

\begin{definition}[Quartet metric]\label{def:quartet_metric}
    Let $T$ be a phylogenetic tree on leaf set $X$ with $n = |X| \geq 4$. For any pair of leaves $x \neq y$ of $X$ let
    \begin{equation}
        d_T(x,y)= \sum_{z,w\neq x,y} 2 \cdot \rho_{xy}(T|_{xyzw})+2n-4,
    \end{equation}
    with $d_T(x,x) = 0$. Then, $d_T$ is the \emph{quartet metric} of $T$.
\end{definition}

A \emph{metrization} of a tree $T = (V,E)$ is a map $w: E \rightarrow \mathbb{R}_{> 0}$ that assigns a positive weight to each edge. Next, we define the \emph{quartet metrization} of a phylogenetic tree $T$ on $X$, which assigns a weight $w(e)$ to every edge $e$ of $T$. First note that each internal edge $e$ of $T$ determines a partition of $X$ into 4 non-empty blocks ($X_1$, $X_2$, $X_3$, $X_4$) where the split associated to $e$ is $X_1 \cup X_2 | X_3 \cup X_4$ and the splits associated to the 4 adjacent edges all lead to one $X_i$. To obtain the quartet metrization now assign to every internal edge $e$ the weight
\begin{equation}\label{eq:quartet_metrization1}
    w(e) = |X_1||X_2| + |X_3||X_4|.
\end{equation}
Similarly, every pendant edge, incident to a leaf $x$, induces a tripartition $(\{x\} , X_1 , X_2)$ of $X$. To every such edge $e$ we assign the weight
\begin{equation}\label{eq:quartet_metrization2}
w(e) = |X_1||X_2|.
\end{equation}
The following inequality will be useful later in this manuscript.

\begin{proposition}\label{prop:edge_inequality_weight}
    Let $T = (V,E)$ be a phylogenetic tree on leaf set $X$ with $ |X| \geq 4$ and suppose that $T$ is equipped with the quartet metrization~$w : E \rightarrow \mathbb{R}_{\geq 0}$. Then, for every edge $e$ in $T$, we have $w(e) \geq |X| - 2$.
\end{proposition}
\begin{proof} 
First suppose that $e$ is a pendant edge inducing the tripartition $(\{x\} , X_1 , X_2)$ of $X$. Then, $w(e) = |X_1| |X_2|$. The inequality now follows from the fact that $n_1 \cdot n_2 \geq n_1 + n_2 - 1$ if $n_1, n_2 \geq 1$.

Now suppose $e$ is internal inducing the partition $(X_1$, $X_2$, $X_3$, $X_4)$. Then, $w(e) = |X_1||X_2| + |X_3||X_4|$. Again, the inequality follows because $n_1 n_2 + n_3 n_4 \geq \sum_{i=1}^4 n_i - 2$ whenever all $n_i \geq 1$.
\end{proof}

The relationship of the quartet metric defined in \cref{def:quartet_metric} and the quartet metrization is captured in the following.

\begin{theorem}[\cite{rhodes2019topological}]\label{thm:quartet-metric}
    Let $T = (V,E)$ be a phylogenetic tree on leaf set $X$ with $|X| \geq 4$. Then, for every pair $\{x,y\} \subseteq X$, the quartet metric $d_T (x,y)$ is the distance in $T$ between $x$ and $y$ when $T$ is equipped with the quartet metrization $w: E \rightarrow \mathbb{R}_{\geq 0}$.
\end{theorem}

\begin{corollary}\label{cor:quartet_decomposable}
    Let $T$ be a phylogenetic tree on at least four leaves. Then, the quartet metric $d_T$ is circular decomposable and its support is $\Split (T)$.
\end{corollary}

\section{Displayed quartet metric of a phylogenetic network}\label{sec:displ_quartet_metric}

Here we generalize the quartet metric for a tree to semi-directed networks, accounting for the multiplicities of the quartets. 
Although at first glance this generalization seems primarily of theoretical value since it is unclear whether the quartet multiplicities --- and thus the resulting distances --- can be computed directly from data, the results here are linked in the next section to a second generalization where the generalized
distances can be estimated.

Let $N$ be a semi-directed network on four leaves $\{x,y,z,w\}$, i.e., a quarnet, (possibly a tree). Then, the first step in generalizing the quartet metric is to use the set $\T (N)$ of quartet trees displayed by $N$, accounting for their multiplicities, to obtain
\begin{equation}\label{eq:rhotilde}
    \rhotilde_{xy} (N) = \frac{1}{\mu (N)} \sum_{T \in \T (N)} \mu (T, N) \cdot \rho_{xy} (T).
\end{equation}
That is, the $\rhotilde_{xy}$-value for a pair of leaves~$(x,y)$ of a quarnet~$N$ is simply the average $\rho_{xy}$-values for all the quartet trees that $N$ displays, weighted by multiplicity.
If $N$ is a quartet tree, this reduces to~\cref{eq:rho_tree}. In case the quarnet is strictly level-1, we obtain that $\mu (N) =2$ and $\mu (T, N) = 1$, resulting in $\rho$-values of either $1/2$ or $1$. This exactly coincides with the $\rho$-values used for NANUQ \cite{allman2019nanuq}.

\Cref{fig:rhodistance_tilde} illustrates 
pairwise $\rhotilde$-values for several outer-labeled planar, galled quarnets. For a quarnet to be galled, every hybrid node must be parental to a leaf, but otherwise the hybrid nodes may be chosen arbitrarily in these examples, as long as their choice results in a valid semi-directed network. Thus, for the quarnet in the second column of \cref{fig:rhodistance_tilde}, exactly one of $\{x, y, z, w\}$ needs to be a hybrid, whereas for the third column exactly one of $\{x, y\}$ and one of $\{z,w\}$ need to be hybrids. As an example, note that --- independent of the choice of hybrids --- the quarnet in the third column displays the quartet $xy|zw$ with multiplicity~3, the quartet $xz|yw$ with multiplicity~1, and the quartet $xw|yz$ with multiplicity~0, resulting in the shown $\rhotilde$-values of $1 / 4$, $3/ 4$ and $1$.
Because contracting a 2-blob, or contracting either a 3-blob or a 3-cycle within a blob to a node, has no impact on displayed tree topologies, the values of $\rhotilde$ hold slightly more generally.

\begin{figure}[htb]
    \centering
 \includegraphics[width=0.95\textwidth]{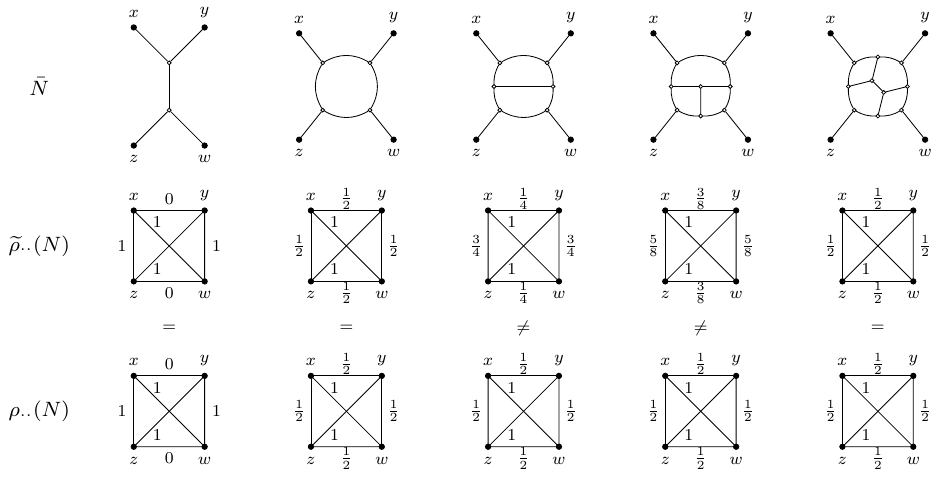}
    \caption{\emph{Top:} Five undirected graphs $\bar{N}$ that can be obtained by undirecting an outer-labeled planar, galled quarnet~$N$ on leaf set $\{x,y,z,w\}$ after contracting 2-blobs, 3-blobs, and 3-cycles within a blob. To make $\bar{N}$ semi-directed again (disregarding the creation of 2-blobs, 3-blobs and 3-cycles), hybrid nodes can be chosen only from the articulation nodes of the undirected blobs since $N$ must be galled, as long as $N$ remains semi-directed. \emph{Middle:} The values $\rhotilde_{..} (N)$ (as defined in \cref{eq:rhotilde}) in case $N$ is outer-labeled planar, galled and $\bar{N}$ is as in the top row are shown on edges connecting two taxa. \emph{Bottom:} The values $\rho_{..} (N)$ (as defined in \cref{eq:rho}) in case $N$ is outer-labeled planar, galled and $\bar{N}$ is as in the top row.}
    \label{fig:rhodistance_tilde}
\end{figure}

As was the case for phylogenetic trees, this newly defined $\rhotilde$ gives rise to a metric on the leaf set of an $n$-leaf semi-directed network. Recall that $N|_{xyzw}$ is the induced quarnet on the leaf set $\{x,y,z,w\}$ of a semi-directed network $N$.
\begin{definition}[Displayed quartet metric]\label{def:displayed_quartet_metric}
    Let $N$ be a semi-directed network on leaf set $X$ with $n = |X| \geq 4$. For any pair of leaves $x \neq y$ of $X$ we let
    \begin{equation}
        \dtilde _N(x,y)= \sum_{z,w\neq x,y} 2 \cdot \rhotilde_{xy}(N|_{xyzw})+2n-4,
    \end{equation}
    with $\dtilde_N(x,x) = 0$. Then, $\dtilde_N$ is the \emph{displayed quartet metric} of $N$.
\end{definition}

Recall that the $\rhotilde$-value of a quarnet is the average over the $\rhotilde$-values of each quartet it displays (\cref{eq:rhotilde}). The following now generalizes Lemma~21 from \cite{allman2019nanuq}. Informally, it says that the $\rhotilde$-value of a quarnet of a network $N$ can equivalently be viewed as a weighted average over $N$'s displayed trees of the $\rho$-values of their induced quartets.

\begin{lemma}\label{lem:dqm}
    Let $N$ be a semi-directed network on leaf set $X$ with $|X| \geq 4$. If $\{x,y,z,w\} \subseteq X$, then
    $$\rhotilde_{xy} (N|_{xyzw}) = \frac{1}{\m (N)} \sum_{T \in \T (N)}  \m (T, N) \cdot \rho_{xy} (T|_{xyzw}) . $$ 
\end{lemma}
\begin{proof} 

Suppose that $N$ has $r$ hybrid nodes and that $N|_{xyzw}$ has $p\leq r$ hybrid nodes. Then, $\mu(N|_{xyzw}) = 2^p$ and $\mu (N) = 2^r$, so that $$\mu(N|_{xyzw}) = 2^{p-r} \mu(N),$$
and for any quartet tree $T'$ on $\{x,y,z,w\}$, $$\mu(T',N|_{xyzw}) = 2^{p-r} \sum_{\substack{T\in\T(N):\\ T|_{xyzw}=T'}} \mu(T,N).$$
Thus,
\begin{align*}
    \rhotilde_{xy} (N|_{xyzw}) &= \frac{1}{\mu (N|_{xyzw})} \sum_{T' \in \T (N|_{xyzw})} \mu (T', N|_{xyzw}) \cdot \rho_{xy} (T')\\
    &=\frac{1}{2^{p-r}\mu (N)} \sum_{T'\in \T (N|_{xyzw})} 2^{p-r}\sum_{\substack{T\in\T(N):\\ T|_{xyzw}=T' }}\mu (T, N) \cdot \rho_{xy} (T|_{xyzw})\\
    &=\frac{1}{ \mu (N)}  \sum_{T\in\T(N)}\mu (T, N) \cdot \rho_{xy} (T|_{xyzw}).
\end{align*}  
\end{proof}

The next theorem generalizes Theorem 23 of \cite{allman2019nanuq}, showing that the displayed quartet metric on a network is simply a weighted sum of the metrics on its displayed trees.
\begin{theorem}\label{thm:displ_quartet_metric_sum}
    Let $N$ be a semi-directed network on leaf set $X$ with $|X| \geq 4$ and let $\{x,y\} \subseteq X$. Then,
    $$\dtilde_N(x,y) = \frac{1}{\m (N)} \sum_{T \in \T(N)} \m (T, N)\cdot d_T (x,y) .$$
\end{theorem}
\begin{proof}
Let $n = |X|$. Using \cref{def:quartet_metric,def:displayed_quartet_metric}, \cref{lem:dqm}, and the equality $\mu (N) = \sum_{T \in \T (N)} \mu (T, N)$, we obtain
\begin{align*}
    \dtilde_N(x,y) &\eqover{\ref{def:displayed_quartet_metric}} \sum_{z,w\neq x,y} 2 \cdot \rhotilde_{xy}(N|_{xyzw})+2n-4 \\
    &\eqover{\ref{lem:dqm}} \sum_{z,w\neq x,y} 2 \cdot \left[ \frac{1}{\mu (N)} \sum_{T \in \T (N)} \mu (T, N)\cdot \rho_{xy} (T|_{xyzw}) \right]+2n-4 \\
    &= \frac{1}{\mu (N)} \sum_{T \in \T (N)} \mu (T,N)\ \left[ \sum_{z,w\neq x,y} 2   \rho_{xy} (T|_{xyzw})  \right] +2n - 4\\
    &\eqover{\ref{def:quartet_metric}} \frac{1}{\mu (N)} \sum_{T \in \T (N)} \mu (T, N)\cdot d_T (x,y).
\end{align*}
\end{proof}

This theorem shows that the metric $\dtilde_N$ is a weighted sum of quartet metrics, with nonnegative weights. Recall that from \cref{cor:quartet_decomposable}, each of these metrics is circular
decomposable and has the induced splits of the corresponding tree as support. Together with the fact that $\Split (\T (N))$ is a circular split system if $N$ is outer-labeled planar (\cref{prop:splits_are_circular}) and by the uniqueness of the
split decomposition of a circular decomposable metric \cite{bandelt1992canonical}, we obtain the following result.

\begin{corollary}\label{cor:displ_quartet_circular}
    Let $N$ be a semi-directed, outer-labeled planar network on at least four leaves. Then, the displayed quartet metric $\dtilde_N$ is circular decomposable, with support $\Split (\T (N))$.
\end{corollary}

\section{NANUQ metric of a phylogenetic network}\label{sec:NANUQ_metric}
The displayed quartet metric $\dtilde_N$ of a network~$N$ presented in the last section has a straightforward connection to the quartet metrics of displayed trees. However, one of its components, the multiplicity of a displayed quartet, is not likely to be obtainable from data. In this section, we propose a different network generalization of the quartet tree metric that does not account for these multiplicities.

First, we introduce a different generalization of $\rho$ from \cref{subsec:quartet_metric}. Let $N$ be a semi-directed network on four leaves $\{x,y,z,w\}$ (which could be a quartet tree), i.e., $N$ is a \emph{quarnet}. Then, in contrast to $\rhotilde$ of the previous section, we ignore multiplicities and instead define
\begin{equation}\label{eq:rho}
    \rho_{xy} (N) = \frac{1}{|\T(N)|} \sum_{T \in \T (N)} \rho_{xy} (T).
\end{equation}
That is, $\rho$ for a pair of leaves of a quarnet is the \emph{unweighted} average $\rho$ of all the quartet trees 
$N$ displays. In case $N$ is strictly level-1, $|\T(N)| = 2$, and so this new extension coincides exactly with the $\rhotilde$-value defined in \cref{eq:rhotilde}, and hence also with the $\rho$ that is used for NANUQ \cite{allman2019nanuq}. For other levels, this newly defined $\rho$ and the $\rhotilde$ from \cref{eq:rhotilde} are generally different (see \Cref{fig:rhodistance_tilde}).

As before, $\rho$ induces a corresponding metric on the leaves of a semi-directed network with four or more leaves.
Since this metric will be the main generalization of the metric in \cite{allman2019nanuq}, we call it the \emph{NANUQ metric}.
\begin{definition}[NANUQ metric]\label{def:NANUQ_metric}
    Let $N$ be a semi-directed network on leaf set $X$ with $n = |X| \geq 4$. For any pair of leaves $x \neq y$ of $X$, let
    \begin{equation}
        d _N(x,y)= \sum_{z,w\neq x,y} 2 \cdot \rho_{xy}(N|_{xyzw})+2n-4,
    \end{equation}
    with $d_N(x,x) = 0$. Then, $d_N$ is the \emph{NANUQ metric} of $N$.
\end{definition}

An immediate consequence of the definition and \Cref{thm:displ_quartet_metric_sum} is the following, expressing the NANUQ metric as a sum of quartet metrics of the displayed trees of the network, plus an error term.
\begin{lemma}\label{lem:NANUQ_error_general}
    Let $N$ be a semi-directed network on leaf set $X$ with $|X| \geq 4$ and let $\{x,y\} \subseteq X$. Then,
    $$d_N(x,y) = \frac{1}{\mu (N)} \sum_{T \in \T(N)} d_T (x,y) \cdot \mu (T, N) + \varepsilon (x,y),$$
    where $\varepsilon (x,y) = 2 \cdot \sum_{z,w \neq x,y} ( \rho_{xy} (N|_{xyzw}) - \rhotilde_{xy} (N|_{xyzw}) ).$
\end{lemma}
\begin{proof}
Let $n = |X|$. Then,
    \begin{align*}
        d_N(x,y) &= \sum_{z,w\neq x,y} 2 \cdot \rho_{xy}(N|_{xyzw})+2n-4 \\
        &= \sum_{z,w\neq x,y} 2 \cdot ( \rhotilde_{xy}(N|_{xyzw}) 
 + \rho_{xy}(N|_{xyzw})- \rhotilde_{xy}(N|_{xyzw}))+2n-4 \\
 &= \dtilde_N (x,y) + \sum_{z,w\neq x,y} 2 \cdot ( \rho_{xy}(N|_{xyzw})- \rhotilde_{xy}(N|_{xyzw})) \\
 &= \frac{1}{\mu (N)} \sum_{T \in \T(N)} d_T (x,y) \cdot \mu (T, N) + 2 \sum_{z,w\neq x,y} ( \rho_{xy}(N|_{xyzw})- \rhotilde_{xy}(N|_{xyzw})) \\
 &= \frac{1}{\mu (N)} \sum_{T \in \T(N)} d_T (x,y) \cdot \mu (T, N) + \varepsilon (x,y).
    \end{align*}
\end{proof}

\subsection{NANUQ metric is circular decomposable for level-2 bloblet networks}
In the rest of this section, we restrict to networks in the class $\BB_2$ with at least four leaves and show that the NANUQ metric is circular decomposable for this class. Recall that this class contains all semi-directed, level-2, outer-labeled planar, galled, bloblet networks. Furthermore, $\BB'_2$ contains all strictly level-2 networks in this class (see \Cref{fig:level2network} for a general example of a network in $\BB'_2$). 
Note that parallel edges (2-cycles) only appear as a 2-blob for a network in $\NN_2$, and since a $k$-taxon bloblet in $\BB_2 \subseteq \NN_2$ with $k \ge 3$ cannot have a 2-blob, we do not need to consider parallel edges for such bloblets.}

For convenience, we assign a standard partition of the leaves to the networks in $\BB'_{2}$. As  shown in \Cref{fig:level2network}, the leaf set $X$ of such a network can be partitioned as $\PP = \{A_1,  A_2 ,B_1,B_2,C_1,C_2\}$,  uniquely up to the interchange of A with B and 1 with 2. We call $\PP$ a \emph{canonical partition} of $X$. Here, we write $A = A_1 \cup A_2$, $B = B_1 \cup B_2$ and $C = C_1 \cup C_2$, where $C$ contains the two hybrids of $N$. 

We say that a set of leaves $S$ is \emph{on cycle~1} (resp.~\emph{on cycle~2}) of $N$ if $S \subseteq A_1 \cup B_1 \cup C_1$ (resp.~$S \subseteq A_2 \cup B_2 \cup C_2$), and that the set $S$ is \emph{on the same cycle} if either of these holds. 

\begin{figure}[htb]
    \centering
    \includegraphics[width=0.35\textwidth]{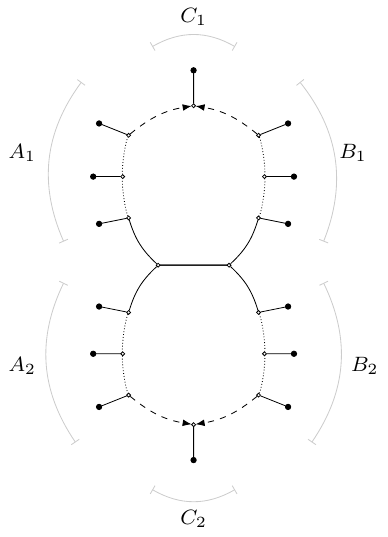}
    \caption{A general strictly level-2 network from the class $\BB'_2$ with a canonical partition $\PP = \{A_1,  A_2 ,B_1,B_2,C_1,C_2\}$ of its (unlabeled) leaves.}
    \label{fig:level2network}
\end{figure}

The following lemma presents an exact expression for the error term $\varepsilon$ of \Cref{lem:NANUQ_error_general} for a network in~$\BB'_2$. In stating it, indices $i$ are taken modulo 2, so, for example, when $i=2$ the set $A_{i+1}$ is $A_1$.

\begin{lemma}\label{lem:NANUQ_exact_error}
    Let $N$ be a semi-directed network on leaf set $X$ from the class $\BB'_{2}$ with $ |X| \geq 4$ and let $\PP$ be a canonical partition of $X$. Let $\{x,y\}$ be any pair of leaves from $X$. Then, 
    $$d_N(x,y) = \frac{1}{4} \sum_{T \in \T(N)} d_T (x,y)\cdot \mu (T, N) + \varepsilon (x,y),$$
    where $\varepsilon (x,y) = \varepsilon (y,x)$ and 
    \begin{align}
         \varepsilon(x,y)  = \frac{1}{2} \cdot 
        \begin{cases}
        - |A_1||A_2| - |B_1||B_2| &\text{if } x \in C_i \text{ and } y \in C_{i+1},\\
        |A_{i+1}| + |B_{i+1}| &\text{if } x \in C_i \text{ and } y \in A_i \cup B_i,\\
        - |B_{i}| &\text{if } x \in C_i \text{ and } y \in A_{i+1},\\
        - |A_{i}| &\text{if } x \in C_i \text{ and } y \in B_{i+1},\\
        - 1 &\text{if } x \in A_i \text{ and } y \in A_{i+1}, \text{or  } x \in B_i \text{ and } y \in B_{i+1}, \\
        0 &\text{otherwise}.\\
        \end{cases}\label{eq:epsform}
    \end{align}
\end{lemma}
\begin{proof}
Since $\mu (N) = 2^2 = 4$, by \Cref{lem:NANUQ_error_general} all that remains is to compute the error term $\varepsilon$, where
$$\varepsilon (x,y) = 2 \sum_{z,w \neq x,y} 
    (\rho_{xy} (N|_{xyzw}) - \rhotilde_{xy} (N|_{xyzw}) 
    ).$$ Next, note that for any $N|_{xyzw}$ its underlying undirected network~$\bar{N} = \overline{N|_{xyzw}}$ with 3-cycles contracted appears --- up to relabeling the leaves --- in one of the three leftmost columns of the top row of \Cref{fig:rhodistance_tilde}. Specifically, if $N|_{xyzw}$ contains no 4-blob, $\bar{N}$ will be the quartet tree in the far-left column, if $N|_{xyzw}$ has a strict level-1 4-blob or a strict level-2 4-blob with a 3-cycle, $\bar{N}$ appears in the second column; otherwise, it appears in the third column. \Cref{fig:rhodistance_tilde} also gives the value of $\rho_{xy} (N|_{xyzw}) - \rhotilde_{xy} (N|_{xyzw})$, with this expression being non-zero only when $N|_{xyzw}$ is strictly level-2 after contracting 3-cycles and $x,y$ are adjacent in the induced circular order. In particular, if these conditions hold, then
    $$\rho_{xy} (N|_{xyzw}) - \rhotilde_{xy} (N|_{xyzw}) = \begin{cases} 1/4 &\text{if $x,y$ are on the same cycle,}\\-1/4 &\text{if $x,y$ are on different cycles.}\end{cases}$$
        
    Computing $\varepsilon (x,y)$ is now a matter of checking every quarnet of $N$ on leaves $x$ and $y$ to see if it falls in one of the above two classes contributing a non-zero term. To simplify this, note that the sign of $\varepsilon$ is completely determined by whether $x$ and $y$ are on the same cycle or not. Next, note that a quarnet on $x$ and $y$ can only be strictly level-2 after contracting 3-cycles, if (i) the quarnet has the two hybrids in its leaf set (otherwise it would be at most level-1); and (ii) it has two leaves on each cycle (otherwise one of the cycles would be a 3-cycle which would be contracted). Furthermore, recall from above that (iii) the term is non-zero only when $x$ and $y$ are adjacent in the quarnet order.
    
    Thus, to determine $\varepsilon (x,y)$ it suffices to first count the number of choices of the other two leaves $z$ and $w$ such that the induced quarnet satisfies the conditions (i)-(iii). Multiplying this count by $\frac{1}{4} \cdot 2 = \frac{1}{2}$, and then determining the sign as above yields the value.

    \emph{Case 1: $x \in C_i$ and $y \in C_{i+1}$}. Condition (i) is already satisfied. To adhere to condition (iii), we can either pick both leaves $z,w$ in $A$ or both in $B$. To adhere to condition (ii), $z,w$ must be on different cycles. In total, this gives $|A_1||A_2| + |B_1||B_2|$ choices of $z$ and $w$.

    \emph{Case 2:  $x \in C_i$ and $y \in A_i \cup B_i$.} By condition (i), we must pick the other hybrid in $C_{i+1}$. Then, to make $x$ and $y$ adjacent (condition (iii)) and to force them to have two leaves on either cycle (condition (ii)), there are exactly $|A_{i+1}| + |B_{i+1}|$ choices.

    \emph{Case 3: $x \in C_i$ and $y \in A_{i+1}$.} Again, we need to pick the other hybrid in $C_{i+1}$ for condition (i). Then, there are $|B_i|$ choices for the fourth leaf to satisfy the other two conditions. 

    \emph{Case 4: $x \in C_i$ and $y \in B_{i+1}$.} Analogous to Case 3.

    \emph{Case 5: $x\in A_i$ and $y \in A_{i+1}$ or $x\in B_i$ and $y \in B_{i+1}$.} In this case, we need to pick the two hybrids in $C$ to satisfy all conditions. Hence, there is only one choice.  
    
    \emph{Case 6: otherwise}. In this case, we have that for some $i \in \{1,2\}$ (1): $\{x, y\} \subseteq A_i$, (2): $\{x,y\} \subseteq B_i$, (3): $x \in A_i$ and $y \in B_i$ (or vice versa), or (4): $x \in A_i$ and $y \in B_{i+1}$ (or vice versa). In all cases, we need to pick the two hybrids. But then, either condition (ii) or condition (iii) is not satisfied. 
    Thus, there are no valid choices.
\end{proof}

Having expressed the NANUQ metric for networks in $\BB_2'$ as a sum of the circular decomposable quartet metric and an explicit error term, we can now show that the NANUQ metric is itself circular decomposable. As a first step we compute bounds on the split weights $\alpha_{ij}$ (see \cref{prop:split_weights}).

Given a bloblet  $N$ in $\BB'_2$ with leaf set $X$ and canonical partition $\PP$, each of its displayed trees can be obtained by removing one hybrid edge for each of the two hybrids. Allowing for multiplicities, this gives rise to four displayed trees (see \cref{fig:trees_displayed}). We denote these trees by $T_{\Sigma_1 \Sigma_2}$ with $\Sigma_i \in \{A_i, B_i\}$, 
where, for example, $\Sigma_i = A_i$ means that the hybrid edge incident to $C_i$ and joined to the subgraph on $A_i$ is kept intact, but its partner hybrid edge (attached to the $B_i$ subgraph) is removed.
Hence, when taking multiplicities into account, every network $N$ then has the displayed trees $T_{A_1 A_2}$, $T_{A_1 B_2}$, $T_{B_1 A_2}$ and $T_{B_1 B_2}$ (see again \Cref{fig:trees_displayed}). Here, multiplicities arise if some of the sets $A_i,B_i$ are empty. For example, the trees $T_{A_1 A_2}$ and $T_{B_1 A_2}$ are isomorphic if $A_1$ and $B_1$ are empty.

\begin{figure}[htb]
    \centering
    \includegraphics[width=\textwidth]{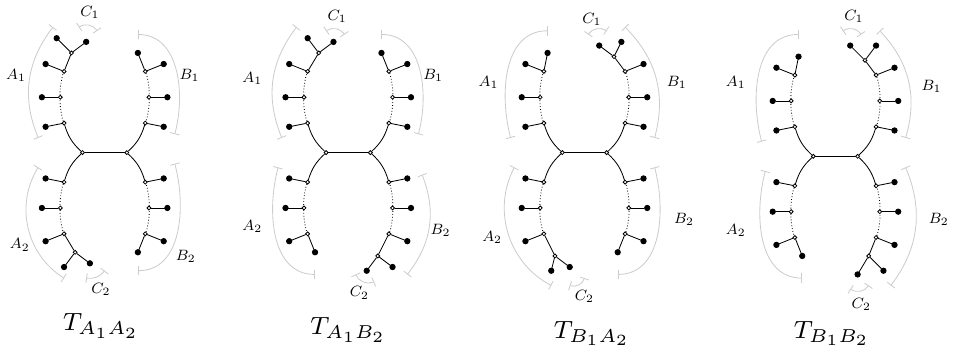}
    \caption{The four trees displayed by a general semi-directed network from the class $\BB_2'$, using the labeling from \Cref{fig:level2network}. If some of the taxon sets $A_i,B_i$ are empty, some of these trees may be isomorphic and have a different topology.}
    \label{fig:trees_displayed}
\end{figure}

We call a non-trivial split of a displayed tree $T$ a \emph{major split} if it separates the leaves in $A$ from those in $B$. Hence, a major split is of the form $A|B \cup C$, $A \cup C | B$, $A \cup C_1 | B \cup C_2$, or $A \cup C_2 | B \cup C_1$. As an example, if $|A|, |B| \geq 1$, then $A \cup C_1 | B \cup C_2$ is a major split of $T_{A_1 B_2}$.

\begin{lemma}\label{lem:splitweights_displayed_trees}
Let $N$ be a network on leaf set $X$ from the class $\BB'_{2}$ with $n=|X| \geq 4$, $\PP$ a canonical partition of~$X$, and $\CC = (x_0, \dots, x_n = x_0)$ the circular order induced by~$N$. Let $T = T_{\Sigma_1 \Sigma_2}$ be a displayed tree of $N$, $d_T$ the quartet metric on $T$, and $S_{ij} \in \SSS (\CC)$ a split. Then, 
\begin{enumerate}[label={(\alph*)},noitemsep]
\item $\alpha_{ij} (d_T) = 0$, if $S_{ij} \not \in \Split (T)$;
\item $\alpha_{ij} (d_T) \geq 2|X| - 4 > 0$, if $S_{ij} \in \Split (T)$;
\item $\alpha_{ij} (d_T) > 2 |A_1||A_2| + 2|B_1||B_2|$, if $S_{ij} \in \Split (T)$ and $S_{ij}$ is a major split of $T$.
\end{enumerate}
\end{lemma}
\begin{proof}
Note that $\CC$ is an induced order for $T$, so  $\Split (T) \subseteq \SSS (\CC)$. By \Cref{cor:quartet_decomposable}, $d_T$ is circular decomposable and its support is $\Split (T)$. Statement (a) then follows from \cref{prop:split_weights}.

To prove statements (b) and (c), let $S_{ij} \in \Split (T)$ be arbitrary and  $e$  the unique edge of $T$ inducing it. Again, by \cref{prop:split_weights} and \Cref{cor:quartet_decomposable}, we already have $$\alpha_{ij} (d_T) = d_T (x_i,x_j) + d_T (x_{i+1}, x_{j+1}) - d_T (x_i, x_{j+1}) - d_T (x_{i+1}, x_j)>0.$$  
First suppose that $i+1 \neq j$ and $j+1 \neq i$. Then, since $S_{ij}$ is the split induced by $e$, $ x_{j+1} x_i | x_{i+1} x_j $ is a quartet induced by $T$ whose internal edge is $e$.
By the 4-point condition for additive metrics on trees \cite{buneman1970-4pointcondition},
$$\alpha_{ij} (d_T) = d_T (x_i,x_j) + d_T (x_{i+1}, x_{j+1}) - d_T (x_i, x_{j+1}) - d_T (x_{i+1}, x_j)=2\cdot w(e),$$
where $w(e)$ is the length of $e$ under the quartet metrization.
If $i+1 = j$ or $j+1 = i$, we similarly obtain that $\alpha_{ij} = 2 \cdot w(e)$.
From \Cref{prop:edge_inequality_weight}, we  have that $w(e) \geq |X| - 2$ and (b) follows. 

Now suppose that $S_{ij}$ is a major split of $T$. Using \cref{eq:quartet_metrization1}, even if some of the sets $A_1, A_2, B_1, B_2$ are empty, we then always have that $w(e) > |A_1||A_2| + |B_1||B_2|$ and the bound in (c) follows.
\end{proof}

The next lemma shows that if $N$ is a network from the class $\BB_2'$ inducing the circular order $\CC$ and $S_{ij} \in \SSS (\CC)$, then, if $S_{ij}$ is not induced by any displayed tree of $N$, its split weight will be zero under $d_N$.

\begin{lemma}\label{lemma:NANUQ_weights_zero}
Let $N$ be a network on leaf set $X$ from the class $\BB'_2$ with $n=|X| \geq 4$,  
$\CC = (x_0, \dots x_n = x_0)$ the circular order induced by~$N$, and $S_{ij} \in \SSS (\CC)$. If $S_{ij} \not \in \Split (\T (N))$, then $\alpha_{ij} (d_N) = 0$.
\end{lemma}
\begin{proof}
Recall that $\mu(N)=4$. For simplicity, we thus show the result for the distances $4 \cdot d_N$. Let $\PP$ be a canonical partition of $X$. Denote by $\TT(N) = \{T_{A_1 A_2}, T_{A_1 B_2}, T_{B_1 A_2}, T_{B_1 B_2} \}$ the multiset of trees displayed by $N$. Then, using the equation from \cref{lem:NANUQ_exact_error} and the fact that $\mu(T, N) = 1$ when summing over all trees~$T$ in the multiset $\TT(N)$, 
$$4 \cdot 
d_N(x,y) = \sum_{T \in \TT(N)} d_T (x,y) + 4\cdot \varepsilon (x,y),$$
with $\varepsilon (x,y)$ as in  \Cref{lem:NANUQ_exact_error}.  
Let $Y = \{ x_i, x_{i+1}, x_j, x_{j+1}\}$ and suppose $S_{ij} \not\in \Split (\T (N))$. Then $S_{ij}$ is a non-trivial split, and thus $x_{i+1} \neq x_j$ and $x_{j+1} \neq x_i$. Hence, $|Y| = 4$. 

As shorthand notation, let $d_{kl}=4 \cdot d_N (x_k, x_l)$ and $\varepsilon_{kl}=4 \cdot \varepsilon (x_k, x_l)$. 
Similarly, for a tree $T \in \TT (N)$, let $d'_{kl} (T)= d_{T} (x_k,x_l)$, and $d'_{kl} = \sum_{T \in \TT (N)} d'_{kl} (T)$. Then \Cref{lem:NANUQ_exact_error} states that $d_{kl} = d'_{kl} + \varepsilon_{kl}$. 
Lastly, we write $D_1 = d_{i,j} + d_{i+1, j+1}$, $D_2 = d_{i, j+1} + d_{i+1, j}$, 
so that $D_1$ and $D_2$ correspond to the positive and negative terms 
for $\alpha_{ij}$ from \cref{prop:split_weights}. Similarly, we write $D'_p$, $D'_p (T)$, and $E_p$ with $p \in \{1,2\}$ for 
analogs of this expression using $d'$, $d' (T)$, and $\varepsilon$, respectively. 

In this notation, $$\alpha_{ij} (4\cdot d_N) = D_1 - D_2 =(D_1'-D_2')+(E_1-E_2).$$
Since $S_{ij} \notin \Split(\T(N))$ and $\alpha_{ij} (d_{T}) = D'_1 (T) - D'_2 (T)$, \Cref{lem:splitweights_displayed_trees}(a) implies that $D_1'(T)-D_2'(T)=0$ for all $T \in \TT (N)$, and hence $D'_1 - D'_2 = 0$. It is therefore enough to show
$E_1=E_2$.
To do so, we consider three cases.

\textbf{Case  $|Y \cap C| = 0$:} Since all leaves in $Y$ are non-hybrids and 
$x_i,x_{i+1}$ are neighbors in $\CC$, $\{x_i, x_{i+1} \}$ is a subset of $A$ or of $B$. Similarly, $\{x_j, x_{j+1} \}$ is a subset of $A$ or of $B$. If $\{x_i, x_{i+1} \} \subseteq A$ and $\{x_j, x_{j+1} \} \subseteq B$ (or vice versa), then from \cref{eq:epsform} we obtain that $E_1 = E_2 = 0$. On the other hand, if $\{x_i, x_{i+1} \} \subseteq A$ and $\{x_j, x_{j+1} \} \subseteq A$ (or both are subsets of $B$), then by \cref{eq:epsform} we get that $E_1 = E_2 \in \{0, -2, -4 \}$. 

\textbf{Case $|Y \cap C| = 1$:}
Without loss of generality, assume $x_i = c_1 \in C_1$ and $x_{i+1} \in B$. 
Note that either $\{x_j, x_{j+1}\} \subseteq A$ or $\{x_j, x_{j+1}\} \subseteq B$, since $x_j,x_{j+1}$ are non-hybrids adjacent in~$\CC$. 
Assume first that $\{x_j,x_{j+1}\}$ is a subset of exactly one of $ A_1, A_2, B_1$, or $B_2$.  Then, since the error terms in \cref{eq:epsform} depend only on the specific set in which the leaves are contained, it follows that $\varepsilon_{i, j} = \varepsilon_{i, j+1}$ and $\varepsilon_{i+1, j+1} = \varepsilon_{i+1, j}$ and that $E_1 = E_2$.

Assume now that $x_j \in B_1$ and $x_{j+1} \in B_2$, and hence $x_{i+1}\in B_1$.  Then $S_{ij}$ is, for instance, a split of $T_{A_1 B_2}$, a contradiction. With the circular order $\mathcal C$, it is not possible that $x_j \in B_2$ and $x_{j+1} \in B_1$, since $x_{i+1}$ is not a hybrid. If $x_j \in A_2$ and $x_{j+1} \in A_1$, then the split $S_{ij}$ is a split of $T_{A_1 B_2}$ (also $T_{A_1 A_2}$), a contradiction. Lastly, since $x_{i+1}\in B$, it is not possible that $x_j \in A_1$ and $x_{j+1} \in A_2$. 

\textbf{Case $|Y \cap C| = 2$:} We will show that in this case the split $S_{ij}$ will always be induced by a displayed tree, and hence this case cannot occur. We may again assume $x_i \in C_1$. 

If $x_j \in C_2$, we may also assume that $x_{i+1} \in B$ and $x_{j+1} \in A$. Then, $S_{ij}$ is a split in the displayed tree $T_{A_1 B_2}$. Similarly, if $x_{j+1} \in C_2$, we may assume that $x_{i+1}, x_j \in B$. Then, $S_{ij}$ is a split in the displayed tree $T_{A_1 A_2}$. Finally, if $x_{i+1} \in C_2$, then either $A$ or $B$ is empty. Assuming the latter, then, $x_j, x_{j+1} \in A$ and $S_{ij}$ is a split in $T_{A_1 A_2}$.
\end{proof}

The next lemma can be seen as a counterpart to \Cref{lemma:NANUQ_weights_zero}. It says that if $N$ is a network from the class $\BB_2'$ inducing the circular order $\CC$ and $S_{ij} \in \SSS (\CC)$, then, if $S_{ij}$ is induced by some displayed tree of $N$, its split weight will be strictly positive under $d_N$.

\begin{lemma}\label{lemma:NANUQ_weights_nonzero}
Let $N$ be a network on leaf set $X$ from the class $\BB'_2$ with $n=|X| \geq 4$, let $\CC = (x_0, \dots, x_n = x_0)$ be the circular order induced by~$N$, and let $S_{ij} \in \SSS (\CC)$. If $S_{ij} \in \Split (\T (N))$, then $\alpha_{ij} (d_N) > 0$.
\end{lemma}

\begin{proof}
Let $\PP$ be a canonical partition of $X$, let $S_{ij} \in \Split (\T (N)) \subseteq \SSS (\CC)$ be arbitrary and let $Y = \{ x_i, x_{i+1}, x_j, x_{j+1}\}$. Note that we have $|Y| \geq 3$ (as $|Y| \leq 2$ does not correspond to a split $S_{ij}$). We again show the result for the distances $4 \cdot d_N$, using the same notation as in the proof of \Cref{lemma:NANUQ_weights_zero}.  We consider four main cases. 

\textbf{Case 1, $|Y \cap C| = 0$ and $|Y|= 4$:}  In this case, all leaves in $Y$ are non-hybrids and thus $|X| \geq 6$. Then, by \cref{eq:epsform}, $E_1 \geq -4$ and $E_2 \leq 0$. Hence, $E_1-E_2 \geq -4$. Since $S_{ij} \in \Split (\TT (N))$, there is a tree $T \in \TT (N)$ which induces $S_{ij}$. Then, relying on the definition of $\alpha_{ij}$ from \cref{prop:split_weights}, $\alpha_{ij} (d_T) = D'_1 (T) - D'_2(T) \geq 2|X| - 4$ by part (b) of \cref{lem:splitweights_displayed_trees}, while for all the other trees $T' \in \TT (N)$ by parts (a) and (b) of \cref{lem:splitweights_displayed_trees}, $\alpha_{ij} (d_{T'}) =  D'_1 (T') - D'_2(T') \geq 0$. Thus, $D'_1 - D'_2 \geq 2|X| - 4 \geq 8$. Therefore, $\alpha_{ij} (4 \cdot d_N) =  (D'_1 - D'_2) + (E_1 - E_2) > 0$. 

\textbf{Case 2, $|Y \cap C| = 1$ and $|Y|= 4$:} In this case, there is exactly one hybrid leaf in $Y$, and we may assume $x_i \in C_1$, $x_{i+1} \in B$ and that the circular order $\CC$ goes clockwise in \cref{fig:level2network}. Then $x_j, x_{j+1}$ are both in $B$ or both in $A$, which we consider as subcases.

\emph{Case 2.1,  $x_j, x_{j+1} \in B$:} If $x_{i+1} \in B_2$, then because the circular order $\CC$ goes clockwise in \cref{fig:level2network}, we must have $x_j, x_{j+1} \in B_2$. Then by \cref{eq:epsform}, $E_1 = E_2$, so $D_1 - D_2 = D'_1 - D'_2 > 0$ by \cref{lem:splitweights_displayed_trees}(b). If $x_{i+1}, x_j, x_{j+1}\in B_1$, or if $x_{i+1} \in B_1$ and $x_j, x_{j+1} \in B_2$, we similarly find
$E_1=E_2$ and $D_1 - D_2 > 0$.
Lastly, if $x_{i+1}, x_j \in B_1$ and $x_{j+1} \in B_2$, then $E_1 = 2|A_2| + 2|B_2| - 2$ and $E_2 = -2 |A_1|$, so $E_1 - E_2 = 2(|A_2| + |B_2| + |A_1| - 1) >0$. We again have that $D'_1 - D'_2 > 0$ by \cref{lem:splitweights_displayed_trees}(b), so $D_1 - D_2 > 0$.

\emph{Case 2.2, $x_j, x_{j+1} \in A$:} By \cref{eq:epsform}, $\varepsilon_{i+1, j} = \varepsilon_{i+1, j+1} = 0$,
so $E_1 = \varepsilon_{i,j}$ and $E_2 = \varepsilon_{i, j+1}$. Moreover, we cannot have $x_j, x_{j+1} \in A_2$, because then $S_{ij}$ would not be a split of a tree in $\TT (N)$. If $x_j, x_{j+1} \in A_1$, then $E_1 = E_2$. By \cref{lem:splitweights_displayed_trees}(b) we know that $D'_1 - D'_2 > 0$, so we get that $D_1 - D_2 > 0$. If instead $x_{j} \in A_2$ and $x_{j+1} \in A_1$,  then $E_1 = -2 |B_1|$ and $E_2 = 2|A_2| + 2|B_2|$. So, $E_1 - E_2 = - 2|A_2| -2 |B_1| - 2|B_2| \geq -2|X|+6$. 
Note that in this case $S_{ij}$ is a split in both $T_{A_1 A_2}$ and $T_{A_1 B_2}$. Hence, $D'_1 - D'_2 \geq 4 |X| - 8$ by \Cref{lem:splitweights_displayed_trees}(b). Thus, $D_1 - D_2 > 0$.

\textbf{Case 3, $|Y \cap C| = 2$ and $|Y| = 4$:} Without loss of generality, assume that $x_i \in C_1$ and that the circular order $\CC$ goes clockwise in \cref{fig:level2network}. We consider three cases depending on the location of the other hybrid in the circular order of the 4 taxa. 

\emph{Case 3.1, $x_{i+1} \in C_2$:} For $x_i$ and $x_{i+1}$ to be neighbors in $\CC$, we must have that $|B|= 0$. We now consider two subcases, noting that the case $x_j, x_{j+1} \in A_1$ follows from  Case~3.1a by symmetry.

\begin{itemize}[noitemsep,topsep=0pt]
    \item Case 3.1a, $x_j, x_{j+1} \in A_2$: In this case $\varepsilon_{i,j} = \varepsilon_{i, j+1}$ and $\varepsilon_{i+1, j} = \varepsilon_{i+1, j+1}$, so $E_1 = E_2$. Since $D'_1 - D'_2 > 0$ by \cref{lem:splitweights_displayed_trees}(b), $D_1 - D_2 > 0$.
    \item Case 3.1b, $x_j \in A_2$ and $x_{j+1} \in A_1$: Now, $E_1 = -2|B_1| - 2|B_2| = -2|B| = 0$ and $E_2 = 2|A_1| + 2|B_1| + 2|A_2| + 2|B_2| = 2|A|$ (since $|B| = 0$ in Case~3.1). Hence, $E_1 - E_2 = -2 |A| = -2 |X| + 4$. Since $S_{ij}$ is a split of $T_{A_1 A_2}$, $T_{A_1 B_2}$ and $T_{B_1 A_2}$, by \cref{lem:splitweights_displayed_trees}(b), $D'_1 - D'_2 \geq 6 |X| - 12$. Thus, $D_1 - D_2 = (D'_1 - D'_2) + (E_1 - E_2) \geq 4|X|- 8> 0$.
\end{itemize}
    
\emph{Case 3.2, $x_j \in C_2$:} We consider three subcases. Note that the case where $x_{i+1} \in B_2$ and $x_{j+1} \in A_2$ is analogous to Case~3.2c. 
Furthermore, we have that $E_1 = -2|A_1||A_2| - 2|B_1||B_2|$ in all three subcases.

\begin{itemize}[noitemsep,topsep=0pt]
    \item Case 3.2a, $x_{i+1} \in B_1$ and $x_{j+1} \in A_2$: Thus, $E_2 = -2|B_1| - 2|A_2|$. Next, note that $S_{ij}$ is a major split of $T_{A_1 B_2}$. Thus, by \cref{lem:splitweights_displayed_trees}(c), $D_1'-D_2' \geq D'_1 (T_{A_1 B_2}) - D'_2 (T_{A_1 B_2}) > 2|A_1||A_2| + 2|B_1||B_2|$. Since $E_1 - E_2 \geq E_1 \geq -2|A_1||A_2|  -2|B_1||B_2|$, we have $D_1 - D_2 =(D_1'-D_2') + (E_1-E_2) > 0$.
    \item Case 3.2b, $x_{i+1} \in B_2$ and $x_{j+1} \in A_1$: Since $x_i, x_{i+1}$ and $x_j, x_{j+1}$ are neighbors in $\CC$,  $|A_2| = |B_1| = 0$. 
    Thus $E_1=0$ and since
   $\varepsilon_{j, i+1} = 2|A_1| + 2|B_1| = 2|A_1|$ and $\varepsilon_{j+1, i} = 2|A_2| + 2|B_2| = 2|B_2|$, we find $E_2 = 2|A_1| + 2|B_2|$. Thus $E_1 - E_2 = -2|A_1| - 2|B_2| = -2|X| + 4$. Since $S_{ij}$ is a split of $T_{A_1, A_2}$ and $T_{A_1, B_2}$,  by \Cref{lem:splitweights_displayed_trees}(b), $D'_1 - D'_2 \geq 4 |X| - 8$ and it follows that $D_1 - D_2 > 0$. 
    \item Case 3.2c, $x_{i+1} \in B_1$ and $x_{j+1} \in A_1$: We must have $|A_2| = 0$.  Moreover, $\varepsilon_{j, i+1} = -2|A_2| = 0$ and $\varepsilon_{j+1, i} = 2|A_2| + 2|B_2| = 2|B_2|$, so $E_2 = 2|B_2|$. Using $|A_2| = 0$ again, $E_1 - E_2 = -2|B_1||B_2| - 2|B_2|$. Then, since $S_{ij}$ is a major split of $T_{A_1 A_2}$ and $T_{A_1 B_2}$, we have by \Cref{lem:splitweights_displayed_trees}(c) that $ D'_1 (T_{A_1 A_2}) - D'_2 (T_{A_1 A_2}) > 2|A_1||A_2| + 2|B_1||B_2| = 2|B_1||B_2|$, and similarly $D'_1 (T_{A_1 B_2}) - D'_2 (T_{A_1 B_2}) > 2|B_1||B_2|$. Thus, $D'_1 - D'_2 > 4|B_1||B_2|$ and so $D_1 - D_2 = (D'_1 - D'_2) + (E_1 - E_2) > 2 |B_1| |B_2| - 2 |B_2| \geq 0$. In particular, $D_1 - D_2 > 0$.
\end{itemize}

\emph{Case 3.3, $x_{j+1} \in C_2$:} We consider two subcases. Note that the case with $x_{i+1}, x_j \in B_1$ follows from Case~3.3b by symmetry.

\begin{itemize}[noitemsep,topsep=0pt]
    \item Case 3.3a, $x_{i+1} \in B_1$ and $x_j \in B_2$: Then, $E_1 = -2|A_1| - 2|A_2| = -2|A|$ and $E_2 = -2|A_1||A_2| - 2|B_1||B_2| - 2$. Hence, $E_1 - E_2 \geq E_1 = -2|A| \geq -2 (|X| - 4)$ using that $|B| ,|C| \geq 2$ in this case. Then, since $S_{ij}$ is a split in $T_{A_1 A_2}$, we have by \cref{lem:splitweights_displayed_trees}(b) that $D'_1 - D'_2 \geq 2|X| - 4$. Thus, $D_1 - D_2 > 0$. 
    \item Case 3.3b, $x_{i+1}, x_j \in B_2$: Since $x_{i}$ and $x_{i+1}$ are neighbors in $\CC$, we must have $|B_1| = 0$. We also have $E_1 = -2|A_1| + 2|B_1| + 2|A_1| = 0$ and $E_2 = -2|A_1||A_2| - 2|B_1||B_2| = -2|A_1||A_2|$. Thus, $E_1 - E_2 \geq 0$. Since $D'_1 - D'_2 > 0$ by \cref{lem:splitweights_displayed_trees}(b), it follows that $D_1 - D_2 > 0$. 
\end{itemize}

\textbf{Case 4, $|Y|= 3$:} We may assume $x_{i+1} = x_j$, so $E_1 = \varepsilon_{i,j} + \varepsilon_{j, j+1}$ and $E_2 = \varepsilon_{i, j+1}$. Moreover, since $S_{ij}$ is the trivial split $\{x_j\} \mid \left( X \setminus \{x_j\} \right)$ which appears in all four trees in $\TT (N)$, by \cref{lem:splitweights_displayed_trees}(b), $D'_1 - D'_2 \geq 8|X| - 16$. 
If $|Y \cap C| =0 $, the argument for Case~1 applies (with  $|X| \geq 5$) , so we may henceforth assume $|Y \cap C| \geq 1$. We are left with the following four subcases, up to symmetry.

\emph{Case 4.1, $x_i \in C$ and $x_j, x_{j+1} \not \in C$:} In this case, we have that $E_1 \geq -2 \max \{|A_1|, |A_2|, |B_1|, |B_2| \} -2 \geq -2 (|X| - 2) - 2 =  -2|X| + 2$, while $E_2 \leq 2 \max \{|A_1| + |B_1|, |A_2| + |B_2| \} \leq 2 |X| - 4$. Hence, $E_1 - E_2 \geq -4 |X| +6$. Thus, since  $D'_1 - D'_2 \geq 8|X| - 16$, we obtain $D_1 - D_2 > 0$.

\emph{Case 4.2, $x_j \in C$ and $x_i, x_{j+1} \not \in C$:} In this case we get that $E_1 \geq -4 \max \{|A_1|, |A_2|, |B_1|, |B_2| \} \geq -4 |X| +8$, while $E_2 \leq 0$. Hence, $E_1 - E_2 \geq -4 |X| +8$ and again $D_1 - D_2 > 0$ as in Case~4.1.

\emph{Case 4.3, $x_i, x_j \in C$ and $x_{j+1} \not \in C$:} We may assume that $x_i \in C_1$, $x_{j} = x_{i+1} \in C_2$. Since $x_i$ and $x_{i+1}$ are neighbors in $\CC$, we may also assume $|B| = 0$ and $x_{j+1} \in A_2$. Then, $E_1 = -2 |A_1||A_2| - 2 |B_1||B_2| + 2|A_1| + 2|B_1| = -2|A_1||A_2| + 2|A_1|$. Since $E_2 = -2|B_1| = 0$, we see $E_1 - E_2 = -2 |A_1| |A_2| + 2|A_1|$. 
Now consider the tree $T_{A_1 B_2}$ and let $e$ be the edge in this tree inducing the split $S_{ij} = C_2 | \left(X \setminus C_2\right)$. Then, under the quartet metrization~$w$, we have that $w(e) = (|A_1| + 1) |A_2|$ (see \cref{eq:quartet_metrization2}).
Hence, by a similar argument as in the proof of \cref{lem:splitweights_displayed_trees}(c), we obtain that $D'_1 - D'_2 \geq D'_1 (T_{A_1 B_2}) - D'_2 (T_{A_1 B_2}) >  2|A_1| |A_2|$. Therefore, $D_1 - D_2 = (D'_1 - D'_2) + (E_1 - E_2) > 0$.

\emph{Case 4.4: $x_i, x_{j+1} \in C$ and $x_j \not \in C$.} By symmetry, we may assume that $x_i \in C_1$, $x_j \in B_1$ and $x_{j+1} \in C_2$, so $|B_1| = 1$ and $|B_2| = 0$. Then, $E_1 = 2|A_2| + 2|B_2| -2|A_2| = 2|B_2| = 0$ and $E_2 = -2|A_1||A_2| - 2|B_1||B_2| = -2|A_1||A_2|$. Hence, $E_1 - E_2 \geq 0$, so $D_1 - D_2 > 0$.
\end{proof}

We now show that the NANUQ metric is circular decomposable with support exactly the set of splits induced by the displayed trees of a network. We conjecture that this theorem can be extended to the class~$\NN_2$ (for some more details, see the discussion in \cref{sec:discussion}).
\begin{theorem}\label{thm:NANUQ_circular}
  Let $N$ be a semi-directed network on at least four leaves from the class $\BB_2$. Then, the NANUQ metric $d_N$ is circular decomposable and its support is $\Split (\T (N))$.  
\end{theorem}
\begin{proof}
We may assume that $N$ is in $\BB'_2$, i.e., is strictly  level-2, since the result is known for level-1 networks \cite{allman2019nanuq}. By \cref{prop:splits_are_circular},  $\Split (\T (N) )$ is a circular split system congruent with the circular order induced by~$N$. Then, by \cref{prop:split_weights,lemma:NANUQ_weights_zero,lemma:NANUQ_weights_nonzero}, $d_N$ is circular decomposable and has support $\Split (\T (N))$.
\end{proof}

\section{Identifiability of a canonical form}\label{sec:identifiability}
In this section, we prove identifiability results for a canonical form of networks in $\BB_2$ and $\NN_2$, which can be obtained by applying a series of operations introduced below. Since the canonical form may not itself be a semi-directed network, its main purpose is not to infer the quartet-related properties from previous sections directly --- although quartet information can still be recovered by reversing the operations until a valid semi-directed network is reached. Instead, the canonical form is primarily intended to determine whether a pair of networks is distinguishable. Moreover, it captures the essential features relevant for inference and reflects the kind of output one would expect from an algorithmic method based on our results (see also the discussion in \cref{sec:discussion}).

Recall from \cref{sec:preliminaries} that, given a semi-directed network $N$ on $X$, the semi-directed network obtained by \emph{contracting} its $2$-blobs is the network where every 2-blob in $N$ is replaced by a degree-2 node which is then suppressed (see \cref{fig:canonical_operations}(i)). Similarly, the semi-directed network obtained by \emph{contracting} its $3$-blobs is the network where every 3-blob in $N$ is replaced by a degree-3 node (see \cref{fig:canonical_operations}(ii)). By contracting a 3-cycle in a blob $B$, we mean that the 3-cycle is replaced by a single node (see \cref{fig:canonical_operations}(v)). By \emph{undirecting} a 4-cycle in a semi-directed network, we mean that every directed edge in the 4-cycle is replaced with an undirected edge (see \cref{fig:canonical_operations}(viii)). Since actions such as these result in graphs which are not necessarily semi-directed networks, we refer to these as `mixed graphs'.

\begin{figure}[htb]
    \centering
\includegraphics[width=\textwidth]{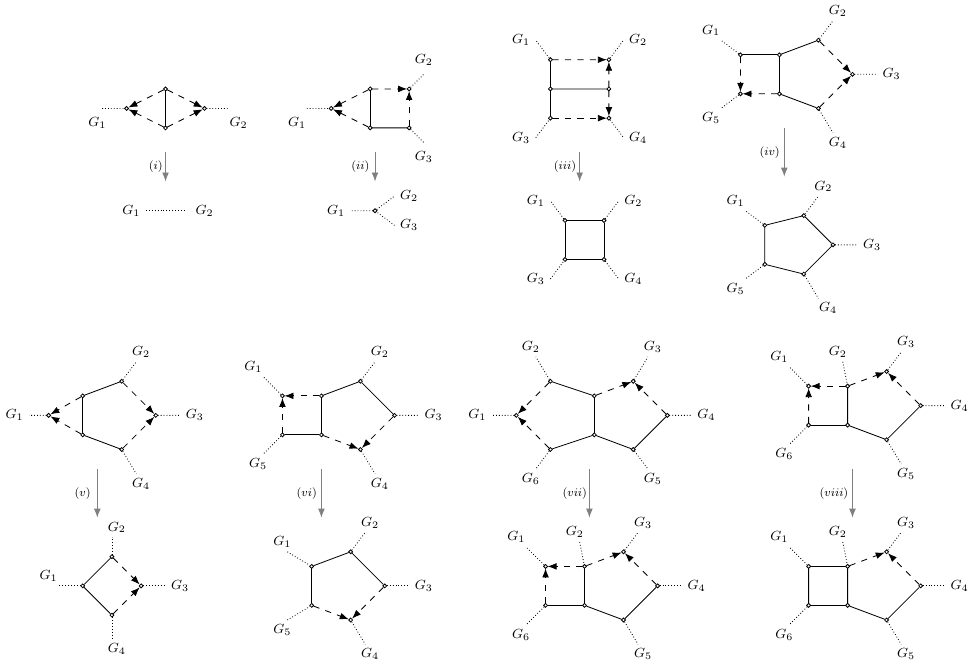}
    \caption{Illustration of operations $(i)$-$(viii)$ used in \cref{def:canonical_form}.}
    \label{fig:canonical_operations}
\end{figure}

If $N$ is outer-labeled planar, each $k$-blob $B$ of $N$ (with $k \geq 4$) induces a unique circular order of the subnetworks around the blob $B$ \cite{rhodes2025identifying}. The mixed graph obtained from $N$ by replacing $B$ by its \emph{representative cycle} is the graph where $B$ is replaced with an undirected $k$-cycle inducing the same circular order of the subnetworks around $B$ (see e.g. \cref{fig:canonical_operations}(iii) and (iv)).

\medskip

Leading up to \cref{def:canonical_form}, we restrict attention to $N$ in $\NN_2$, and $B$ in $\BB'_2$  a $k$-blob, $k\geq 5$, with hybrid nodes $u$ and $v$. If $B$ is a 5-blob that is isomorphic to the blob of \cref{fig:5blobs}(b), then we say $B$ is \emph{split symmetric}. This naming stems from the fact that the set of splits of the displayed trees in such a blob remains the same when cyclically permuting the leaves.

\begin{figure}[htb]
    \centering
\includegraphics[width=0.82\textwidth]{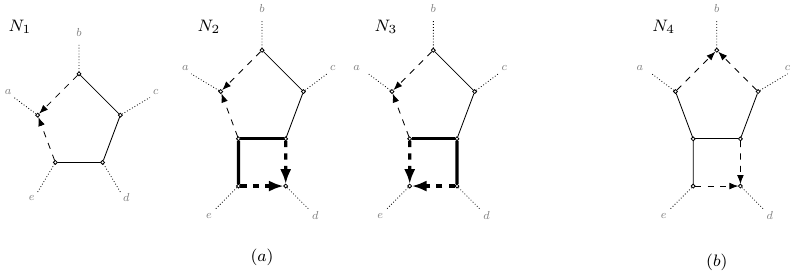}
    \caption{The four possible 3-cycle-free 5-leaf networks in $\BB_2$, up to relabeling the leaves. Each subfigure contains networks with the same canonical form. The networks $N_2$ and $N_3$ both have a dts-undetectable 4-cycle (in thick black), and their canonical form is the network $N_1$. The network $N_4$ consists of the split symmetric 5-blob, with its canonical form obtained by replacing the blob by a completely undirected 5-cycle (see also \cref{fig:canonical_operations}(iv)).}
    \label{fig:5blobs}
\end{figure}

A 4-cycle in $B$ is a \emph{dts-undetectable 4-cycle (displayed-tree-split undetectable 4-cycle)} if it contains a hybrid node $u$ and a parent $p_v$ of the other hybrid node $v$. See \cref{fig:5blobs}(a) for two dts-undetectable 4-cycles in 5-blobs, and \cref{fig:6blobs}(a) for two dts-undetectable 4-cycles in 6-blobs. 
Every dts-undetectable 4-cycle can be labeled as $(p_u, u, p'_u, w)$, where $p_u, p'_u$ are the two parents of~$u$ with $p_u, u$ articulation nodes and one of $\{p'_u, w\}$ a parent of the other hybrid node~$v$. By \emph{suppressing} a dts-undetectable 4-cycle $(p_u, u, p'_u, w)$, we mean deleting the hybrid edge $(p_u, u)$, undirecting the hybrid edge $(p'_u, u)$, and subsequently suppressing the two resulting degree-2 nodes. Suppressing a dts-undetectable 4-cycle turns $B$ into a level-1 blob with the same circular order and with the hybrid $v$ retained as the unique hybrid node in the now-cycle $B$ (see e.g. \cref{fig:canonical_operations}(vi)).

\begin{figure}[htb]
    \centering
\includegraphics[width=\textwidth]{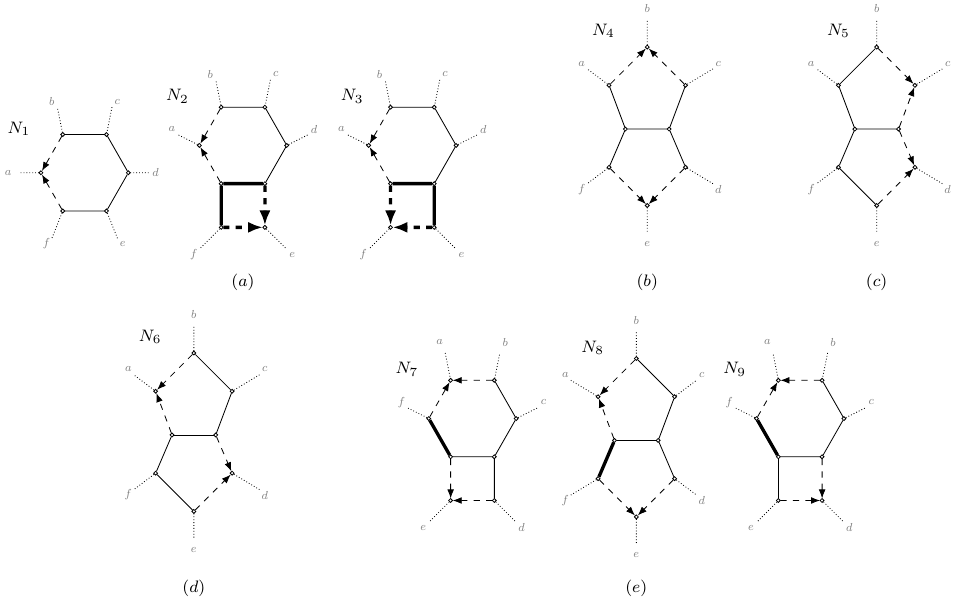}
    \caption{The nine possible 3-cycle-free 6-leaf networks in $\BB_2$, up to relabeling the leaves. Each subfigure contains networks with the same canonical form. The networks $N_2$ and $N_3$ both have a dts-undetectable 4-cycle (in thick black), with canonical form $N_1$. The networks $N_7$ and $N_8$ have a dts-undetectable edge of type~1 (in thick black), and network $N_9$ has a dts-undetectable edge of type~2 (in thick black). The canonical form for the networks in class (e) is non-binary (i.e, the degree of some non-leaf vertex is strictly greater than 3) and obtained by contracting the thick black edges and undirecting the hybrid edges in the remaining 4-cycle within the blob (see also \cref{fig:canonical_operations}(vii) and (viii)). }
    \label{fig:6blobs}
\end{figure}

A non-hybrid edge of $B$ incident to an articulation node is called a \emph{frontier edge}. The frontier edges of $B$ are exactly the non-hybrid edges on the `outside of the blob' in a planar representation, such as in \cref{fig:5blobs,fig:6blobs}. Lastly, suppose that $B$ is not a split symmetric 5-blob and does not have a dts-undetectable 4-cycle (note that then $k\geq 6$). A frontier edge $\{s, t \}$ of $B$ is a \emph{dts-undetectable edge (displayed-tree-split undetectable edge)} if (1) both $s$~and~$t$ are parents of a hybrid in $B$, or (2) $s$ is a parent of a hybrid in $B$ and $t$ is part of a 4-cycle in $B$ (or vice versa). See \cref{fig:6blobs}(e) for examples. Note that since $k\geq 6$, each blob can have at most one dts-undetectable edge.

\medskip

The operations and definitions illustrated in \cref{fig:5blobs,fig:canonical_operations,fig:6blobs},
now permit the following. 

\begin{definition}\label{def:canonical_form}
Let $N$ be a semi-directed network on $X$ from the class $\NN_2$. The \emph{canonical form} $\canonical{N}$ of $N$ is the mixed graph obtained from $N$ by applying the following operations in order:
\begin{enumerate}[label={(\roman*)},noitemsep, topsep=0pt]
    \item contracting every 2-blob;
    \item contracting every 3-blob; 
    \item replacing every 4-blob with its representative cycle;
    \item replacing every split symmetric 5-blob with its representative cycle;
    \item contracting 3-cycles in $k$-blobs ($k \geq 5$);
    \item suppressing dts-undetectable 4-cycles in $k$-blobs ($k \geq 5$);
    \item contracting dts-undetectable edges in $k$-blobs ($k \geq 6$);
    \item undirecting the remaining 4-cycles in $k$-blobs ($k\geq 6$)

\end{enumerate}
\end{definition}

Since operations $(i)$-$(viii)$ only alter the structure of individual blobs, the operations result in the same mixed graph independent of the order in which the blobs are considered. Hence, the canonical form of a network is unique.

The following characterizes those networks whose canonical form matches itself. Note that for 2-, 3- and 4-cycle-free bloblets in $\BB_2'$, the condition that a bloblet is in canonical form is, equivalently, that the two hybrids $x_{h_1}, x_{h_2}$ are not separated by exactly one leaf, $\dots, x_{h_1}, x_l, x_{h_2}, \dots$, in the induced circular order.

\begin{observation}\label{obs:canonical_is_network}
Let $N$ be a 2-, 3-, and 4-cycle-free semi-directed network from the class $\NN_2$ such that the edge distance between two hybrid nodes in any strictly level-2 blob is not 3. Then, $N \cong \canonical{N}$.
\end{observation}

The importance of the canonical form is highlighted in the following.

\begin{lemma}\label{lem:canonical_iff_split}
Let $N_1$ and $N_2$ be two networks from $\NN_2$ on the same leaf set $X$. Then, $\canonical{N_1} \cong \canonical{N_2}$ if and only if $\Split (\T (N_1)) = \Split (\T (N_2))$.
\end{lemma}

\begin{proof} 
First note that operations $(i)$-$(viii)$ from \cref{def:canonical_form} preserve the tree-of-blobs of a network. Since the displayed splits of a network uniquely determine its displayed quartets, and hence its tree-of-blobs \cite{allman2023tree}, it suffices to consider $N_1, N_2$ with the same tree-of-blobs. Hence, we can treat the blobs separately and may assume $N_1, N_2 \in \BB_2$.

The result is straightforward if $n = |X| \in \{2,3,4\}$, so we assume $n\geq 5$ and disregard operations $(i)$-$(iii)$ in \cref{def:canonical_form}. It can be checked directly that two networks that are the same after applying operation~$(iv)$, which applies to a specific 5-taxon network, have the same displayed splits. We next explain that operations $(v)$ and $(vi)$, which 
keep networks in $\BB_2$, have no effect on the set of splits of the trees displayed by the network. 

Operation $(v)$,
contracting a 3-cycle to a node preserves the set of displayed trees, since a 3-cycle displays exactly the same trees as a 3-leaf tree.

If a network has a dts-undetectable 4-cycle, it must have one of two structures. See Figures \ref{fig:5blobs}(a) and \ref{fig:6blobs}(a) for these in 5- and 6-blobs, and note that for $n\geq 7$ the cycle at the `top' simply becomes larger (i.e., has more leaves between $b$ and $c$). The suppression of these 4-cycles by operation $(vi)$ does remove some of the displayed trees, all of which are caterpillars,
so that, for instance, in $N_2$ of \Cref{fig:6blobs} the 
tree $(((f,e),a))...)$ is not displayed anymore.
This only removes the split $fe|a...$. As this split is on other displayed trees before and after operation $(vi)$, it has no impact on the collection of splits of the displayed trees.

Next, consider operations $(vii)$ and $(viii)$ of \cref{def:canonical_form}. We will show that if two networks become isomorphic after applying these operations, they have the same displayed splits. If only operation~$(vii)$ (resp. operation~$(viii)$) is applied, the pair of networks must be isomorphic to the pair of networks in \cref{fig:split_lemma_proof}(a) (resp. \cref{fig:split_lemma_proof}(b)). 
Lastly, if both operations~$(vii)$ and $(viii)$ are applied, the pair of networks comes from the three networks $N_7$, $N_8$, and $N_9$ in \cref{fig:6blobs}(e), with possibly more leaves between $b$ and $c$ if $n>6$. Similar to the previous paragraph, it is straightforward to check that the sets of displayed splits coincide in all three cases.

For the remainder of the proof, 
we may assume that $N_1=N_1^c$ and $N_2=N_2^c$ have the  canonical form produced by applying operations $(iv)$-$(viii)$, so they have no 3-cycles, no dts-undetectable edges, are not symmetric 5-blobs, may be non-binary only as produced by operation $(vii)$, and only have undirected 4-cycles not stemming from dts-undetectable 4-cycles.
That $N_1^c=N_2^c$ implies $\Split (\T ({N_1})) = \Split (\T ({N_2}))$ is now obvious. It remains to show the converse.

\smallskip

Suppose then that $N_1=\canonical{N_1} \not \cong \canonical{N_2}=N_2$. If $N_1,N_2$ induce different circular orders, then by \cite[Thm.\,5.3]{rhodes2025identifying}, the displayed quartets of $N_1$ and $N_2$ differ, and hence $\Split (\T (N_1)) \neq \Split (\T (N_2))$, as required. Thus we may assume $N_1$ and $N_2$ induce the same circular orders henceforth. 

\begin{figure}[htb]
    \centering
\includegraphics[width=.8\textwidth]{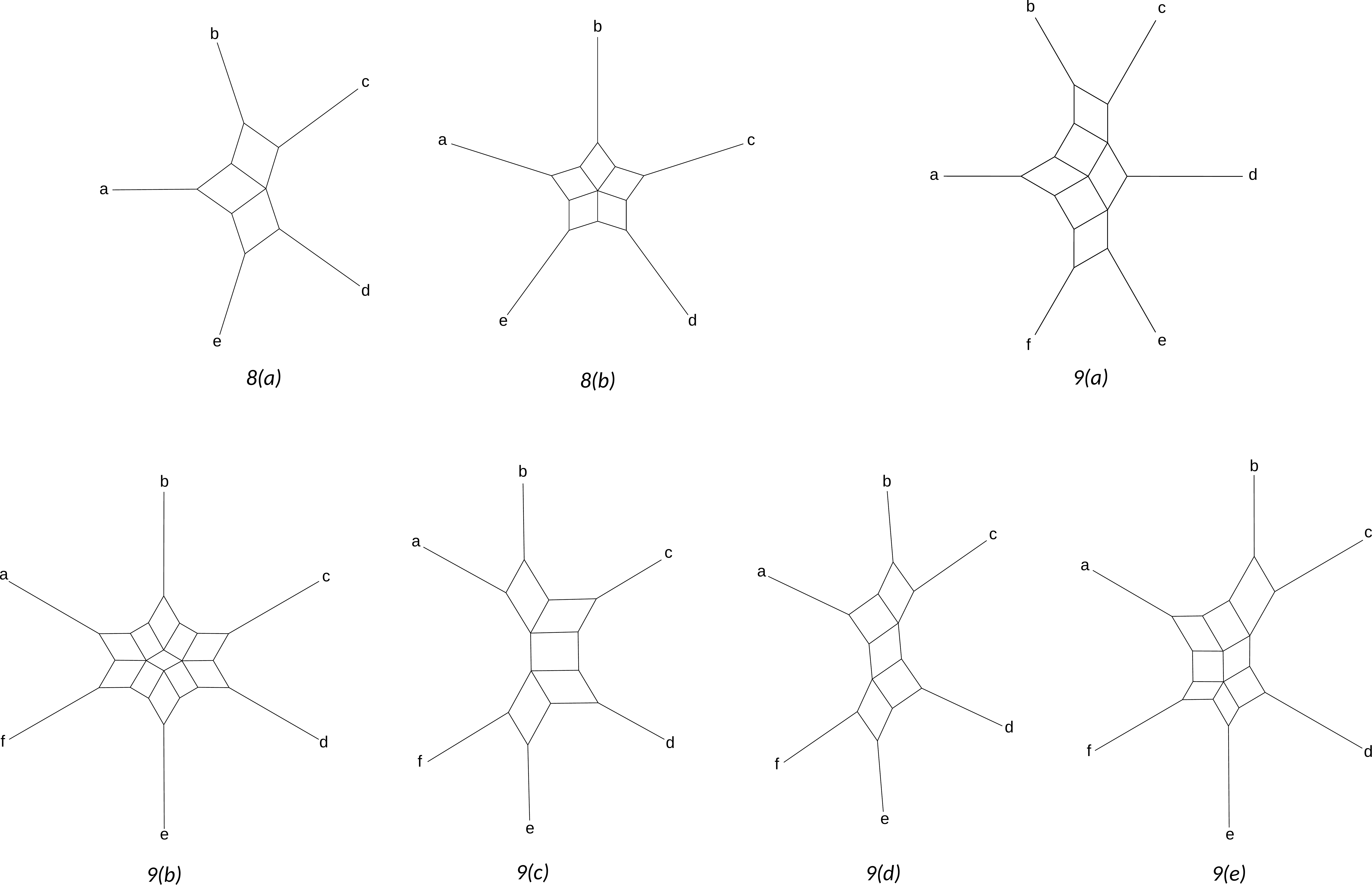}    
\caption{Splits graphs of the displayed splits of the seven sets of networks with different canonical forms in \cref{fig:5blobs,fig:6blobs}.}\label{fig:splitsgraph}
\end{figure}

We will prove that $\Split (\T (N_1)) \neq \Split (\T (N_2))$ by induction on $n$, with base cases $n= 5$ and $n=6$. 
Since $\Split (\T (N_1))$ and $\Split (\T (N_2))$ are both (unweighted) circular split systems, they can be depicted by a splits graph (see e.g. \cite{huson2010phylogenetic} for details). \cref{fig:splitsgraph} shows that the splits graphs of networks from \cref{fig:5blobs,fig:6blobs} (i.e., when $n=5$ or $n=6$) are distinct for different canonical forms. Hence, the base cases hold.

Now suppose $n\geq 7$ and that the theorem holds for smaller $n$. We will repeatedly use that if $\canonical{{{N_1}|_Y}} \not \cong \canonical{{{N_2}|_Y}}$ for some $Y \subset X$, then $\Split (\T ({N_1}|_Y)) \neq \Split (\T ({N_2}|_Y))$ by the induction hypothesis, and hence $\Split (\T ({N_1})) \neq \Split (\T ({N_2}))$. We denote this implication by $(*)$.

\begin{figure}[htb]
    \centering
\includegraphics[width=\textwidth]{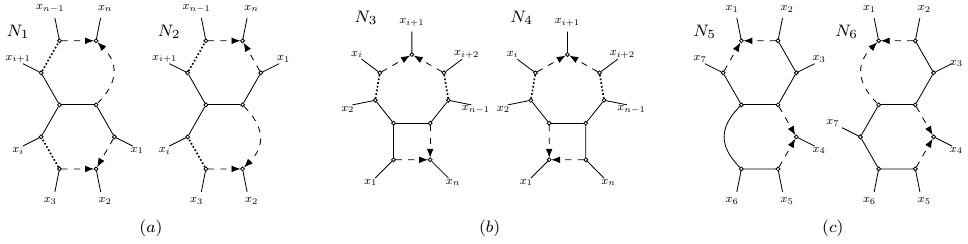}
    \caption{Three pairs of networks from the class $\BB_2$ used in the proof of \cref{lem:canonical_iff_split}. Dotted edges indicate locations where additional leaves may be attached. The networks $N_1$ and $N_2$ in subfigure (a) have $n\geq 7$ leaves and $4 \leq i \leq n-3$. The networks $N_3$ and $N_4$ in subfigure (b) have $n\geq 7$ leaves and $3 \leq i \leq n-4$.}
    \label{fig:split_lemma_proof}
\end{figure}

We first consider the case where one network, say $N_1$, is isomorphic to the network~$N_5$ of \cref{fig:split_lemma_proof}(c). Then, every displayed tree on the network $\canonical{{({N_1}|_{X \setminus \{x_1\}})}}$ (resp. $\canonical{{({N_1}|_{X \setminus \{x_4\}})}}$) induces the non-trivial split $x_2 x_3| x_4 x_5 x_6 x_7$ (resp. $x_5 x_6 | x_1 x_2 x_3 x_7$). 
If $\canonical{{({N_2}|_{X \setminus \{x_1\}}})}$ (resp. $\canonical{{({N_2}|_{X \setminus \{x_4\}})}}$) does not have this split on every displayed tree, we are done by $(*)$. This implies $x_1$ and $x_4$ are hybrids in both networks. Moreover, $x_1,x_2,x_3$ (resp. $x_4,x_5,x_6$) must be on the same cycle
to get these splits. Since $N_1 \not \cong N_2$ but the networks induce the same circular order, we must then have that $N_2$ is isomorphic to network~$N_6$ in \cref{fig:split_lemma_proof}(c).
But then $\canonical{{({N_1}|_Y)}} \not \cong \canonical{{({N_2}|_Y)}}$ for $Y = \{x_1, x_2, x_3, x_4, x_5, x_7\}$ and we are done by $(*)$.

We may now assume that neither $N_1$ nor $N_2$ is isomorphic to $N_5$ in \cref{fig:split_lemma_proof}(c). Then, we claim there exists a set of leaves $\{x,a,b,c\}$ such that (i) $N_1$ (and hence $N_2$) induces the circular order $(x, a, b, c, \ldots )$, (ii) $x$ is a hybrid leaf of $N_1$, (iii) $c$ is not a hybrid leaf of $N_1$, (iv) $x, a, b$ are incident to the same $k$-cycle~$S$ of $N_1$ with $k\geq 6$ (i.e.~at least 4 leaves are incident to this cycle). A set satisfying properties (i)-(iii) must exist since $n\geq 7$, and $N_1, N_2 \not \cong N_5$ ensures (iv) is also possible.

Properties (i)-(iv) ensure that by taking the subnetwork of $N_1$ induced by $X \setminus \{a\}$ (where we write $N'_1 = {N_1}|_{X \setminus \{a\}}$ and $N'_2 = {N_2}|_{X \setminus \{a\}} $), no 3-cycles, dts-undetectable edges, or new 4-cycles are created. Therefore, $\canonical{{(N'_1)}}$ can simply be obtained from $\canonical{N_1}$ by removing $a$ and suppressing the resulting degree-2 node. Equivalently, $\canonical{N_1}$ can be obtained from $\canonical{{(N'_1)}}$ by inserting the leaf $a$ between leaves $x$ and $b$. If $\canonical{{(N'_1)}} \cong \canonical{{(N'_2)}}$, then since $N_1$ and $N_2$ induce the same circular order, $\canonical{N_2}$ can also be obtained from $\canonical{{(N'_2)}}$ by inserting the leaf $a$ between leaves $x$ and $b$. This is a contradiction, since then $\canonical{N_1} \cong \canonical{N_2}$. Hence, we may assume $\canonical{{(N'_1)}} \not \cong \canonical{{(N'_2)}}$. It then follows from $(*)$ that $\Split (\T ({N_1})) \neq \Split (\T ({N_2}))$.
\end{proof}

In the following proposition, we use the notation $\QQ(N)=\QQ(\T(N))$ to denote the set of quartet trees induced by the displayed trees of a network $N$.
\begin{proposition}\label{prop:canonical_equivalences}
Let $N_1$ and $N_2$ be two networks from $\NN_2$ on the same leaf set $X$. Then, the following are equivalent:
\begin{enumerate}[label={(\roman*)},noitemsep, topsep=0pt]
\item $\canonical{N_1} \cong \canonical{N_2}$;
\item $\QQ( N_1) = \QQ (N_2)$;
\item $\Split( \T (N_1)) = \Split (\T (N_2))$.
\end{enumerate}
Furthermore, if $N_1, N_2 \in \BB_2 \subseteq \NN_2$ and $|X| \geq 4$, then $(i)$-$(iii)$ are equivalent to
\begin{enumerate}[noitemsep, topsep=0pt]
\item[(iv)] $d_{N_1} = d_{N_2}$.
\end{enumerate}
\end{proposition}
\begin{proof}
The equivalence between $(i)$ and $(iii)$ was shown as \cref{lem:canonical_iff_split}. That $(iii)$ implies $(ii)$ follows easily, since one can obtain the set of displayed quartets from the set of splits of the displayed trees. To see that $(ii)$ implies $(iii)$, suppose $\QQ( N_1) = \QQ (N_2)$. Then, by \cite{allman2023tree}, $N_1$ and $N_2$ have the same trees-of-blobs $T$, and the splits corresponding to cut edges of $N_1$ and $N_2$ are equal. Thus, we can consider each blob of $N_1$ and $N_2$ separately, and it suffices to assume that $N_1, N_2 \in \BB_2$. If $|X| \leq 3$, $(iii)$ is trivially true. If instead $|X| \geq 4$, $d_{N_1} = d_{N_2}$, since the distances $d_{N_1}$ and $d_{N_2}$ are obtained from the quartets of $N_1$ and $N_2$, and we assumed those sets were equal. By \cref{thm:NANUQ_circular} and because the decomposition of a circular decomposable metric is unique, we obtain $(iii)$. 

The remaining implications follow immediately since if $N_1, N_2 \in \BB_2$ and $|X| \geq 4$, $(ii)$ implies $(iv)$ and $(iv)$ implies $(iii)$. 
\end{proof}

Note that in the previous proposition, $(i)$-$(iii) \Rightarrow (iv)$ is also true if $N_1, N_2 \in \NN_2$, but for the converse $N_1, N_2 \in \BB_2$ is necessary. In particular, we cannot use a similar argument as before, since we have not proved that $d_{N_1} = d_{N_2}$ implies $N_1$ and $N_2$ have the same tree-of-blobs. As mentioned before \cref{thm:NANUQ_circular}, we conjecture that this is true for networks in the class~$\NN_2$. Indeed, proving this extension of \cref{thm:NANUQ_circular} to the class~$\NN_2$ would suffice to obtain the equivalence $(i)$-$(iii) \Leftrightarrow (iv)$ in \cref{prop:canonical_equivalences} for networks in~$\NN_2$, and vice versa.

\smallskip

With a slight reformulation of \cref{prop:canonical_equivalences}, we obtain the main theorem.
\begin{theorem}\label{thm:identifiability}
\begin{enumerate}[label={(\alph*)},noitemsep, topsep=0pt]
\item Let $N_1$ and $N_2$ be two semi-directed, outer-labeled planar, level-2, galled networks. Then, $N_1$ and $N_2$ are distinguishable from their displayed quartets if and only if $N_1$ and $N_2$ have different canonical forms.
\item Let $N_1$ and $N_2$ be two semi-directed, outer-labeled planar, level-2, galled bloblet networks on at least four leaves. Then, $N_1$ and $N_2$ are distinguishable from their pairwise NANUQ distances if and only if $N_1$ and $N_2$ have different canonical forms. 
\end{enumerate}
\end{theorem}

As an instance of part (b) of the theorem,
	the pairwise NANUQ distances for each of the five canonical classes of $\BB_2$
	networks on 6 taxa shown in~\cref{fig:6blobs} are given in~\cref{tab:canonicalNANUQdist}. These, together with distances obtained by permuting taxon labels, each correspond to a unique labeled canonical form, although only 3 of the forms, 9(b), 9(c), 9(d), determine a unique network in $\BB_2$. Nonetheless, from the splits graphs in  \cref{fig:splitsgraph}, the circular order of taxa is immediately seen in all cases, and the splits graphs shapes show that the hybrid node in any $k$-cycle with $k\ge 5$ and no dts-undetectable edge is determined. Note that hybrid identifiability from this distance does not hold for all 5-blobs, since the symmetry of the splits graph for 8(b) gives no hybrid information.
	\begin{table}[h!tb]
	\centering
	{\footnotesize
		\begin{subtable}[t]{0.3\textwidth}
		\begin{tabular}{c|ccccc}
		\toprule
			& $b$ & $c$ & $d$ & $e$ & $f$ \\
			\midrule
			$a$ & 14 & 17 & 18 & 17 & 14 \\
			$b$ &  & 11 & 16 & 19 & 20 \\
			$c$ &  &   & 15 & 18 & 19 \\
			$d$ &  &  &  & 15 & 16 \\
			$e$ &  &  &  &    & 11 \\
		\bottomrule
		\end{tabular}
		\caption{Distances for networks $N_1$, $N_2$ and $N_3$, all with the same canonical form. This class includes the level-1 6-cycle network~$N_1$.}
		\label{tab:dNANUQ_level1}
		\end{subtable}
		\hspace{0.25cm}
			\begin{subtable}[t]{0.3\textwidth}
			\begin{tabular}{c|ccccc}
			\toprule
				& $b$ & $c$ & $d$ & $e$ & $f$ \\
				\midrule
				$a$ &  14 & 18 & 18 & 17 & 13 \\
				$b$ &  &  14 & 17 & 18 & 17  \\
				$c$ &  &   &  13 & 17  & 18 \\
				$d$ &  &  &  & 14 & 18 \\
				$e$ &  &  &  &    & 14 \\
			\bottomrule
			\end{tabular}
			\caption{Distances for $N_4$.}
			\label{tab:N4}
		\end{subtable}
		\hspace{0.25cm}
		\begin{subtable}[t]{0.3\textwidth}
		\begin{tabular}{c|ccccc}
		\toprule
				& $b$ & $c$ & $d$ & $e$ & $f$ \\
				\midrule
				$a$ & 11 & 16 & 18 & 18 & 17 \\
				$b$ &  &  13 & 19 & 19 & 18 \\
				$c$ &  &   & 14 & 19 & 18 \\
				$d$ &  &  &  & 13 & 16 \\
				$e$ &  &  &  &    & 11 \\
				\bottomrule
			\end{tabular}
			\caption{Distances for $N_5$.}
			\label{tab:N5}
		\end{subtable}
		
		\bigskip

		\hspace{0.25cm}
		\begin{subtable}[t]{0.3\textwidth}
			\begin{tabular}{c|ccccc}
				\toprule
				& $b$ & $c$ & $d$ & $e$ & $f$ \\
				\midrule
				$a$ & 13 & 16 & 18 & 17 & 16 \\
				$b$ &  &  11 & 17 & 20 & 19 \\
				$c$ &  &   &  16 & 19 & 18 \\
				$d$ &  &  &  &  13 & 16 \\
				$e$ &  &  &  &    & 11 \\
				\bottomrule
			\end{tabular}
			\caption{Distances for $N_6$.}
			\label{tab:N6}
		\end{subtable}
		\hspace{0.25cm}
		\begin{subtable}[t]{0.3\textwidth}
			\begin{tabular}{c|ccccc}
				\toprule
				& $b$ & $c$ & $d$ & $e$ & $f$ \\
				\midrule
$a$ & 14 & 17 & 18 & 17 & 14 \\
$b$ &    & 11 & 17 & 19 & 19 \\
$c$ &    &    & 16 & 18 & 18 \\
$d$ &    &    &    & 13 & 16\\
$e$ &    &    &     &   & 13\\
				\bottomrule
			\end{tabular}
			\caption{Distances for networks $N_7$, $N_8$ and $N_9$, all with the same canonical form.}
			\label{tab:N789}
		\end{subtable}
        
		\caption{Pairwise NANUQ distances for 6-taxon bloblet networks with the five canonical forms of \cref{fig:6blobs}.
		As proved in \cref{thm:identifiability}, these distances distinguish the five canonical forms.}
    \label{tab:canonicalNANUQdist}
    }
	\end{table}

\medskip
Looking to the future, for data analysis grounded in the main results of \cref{thm:identifiability}, note that when considering outer-labeled-planar networks, the necessary quartet information for applying the theorem is known to be obtainable under several models and data types currently in use for inference. For identifiability from gene trees, these include models of gene tree formation by variations of the Network Multispecies Coalescent (NMSC) process with both independent and common inheritance, and a displayed tree (DT) model in which gene trees are formed on the network without incomplete lineage sorting \cite{rhodes2025identifying}.
Quartet relationships are also identifiable directly from sequence data, assuming a Jukes-Cantor or Kimura-2-Parameter substitution process on the displayed trees of the network \cite{englander2025identifiability}. Although both of these works prove identifiability results beyond that of quartet information, neither suggests a clear path to practical inference. Indeed, although \cref{thm:identifiability} is phrased as an identifiability result, its underlying reliance on the NANUQ distance has the potential to contribute to fast network inference methods.

\section{Discussion} \label{sec:discussion}

In this study, we have focused on the class of outer-labeled planar, galled, level-2, semi-directed networks: a class of phylogenetic networks more general than level-1, which allows some interdependent reticulate events. While this formally means that all blobs are assumed to be outer-labeled planar, galled, and level-2, a careful reading of our proofs shows they in fact permit 2- and 3-blobs not meeting these conditions. Consequently, our results can be applied slightly more generally than the class description initially suggests. We have established that most features of the networks in this class --- captured precisely by a canonical form --- are identifiable from displayed quartets, implying by results elsewhere \cite{englander2025identifiability,rhodes2025identifying} that identifiability holds from different types of biological data under several models of evolution (see the end of \cref{sec:identifiability}).

A second contribution of this work is that our proof is constructive and lays the theoretical foundation for a consistent inference algorithm for the canonical form of a class of level-2 networks under the Network Multispecies Coalescent model, in the spirit of NANUQ and NANUQ$^+$ \cite{allman2019nanuq,allman2024nanuq+}. Notably, this would yield the first algorithm for a subclass of level-2 networks that is known to be statistically consistent for specific models and data types. While both purely combinatorial approaches \cite[e.g.,][]{van2009constructing,van2022algorithm} and computationally cumbersome Bayesian and pseudolikelihood approaches \cite[e.g.,][]{yu2015,zhang2018,wen2018} exist, neither addresses potential non-identifiability issues. The specific algorithm, together with its implementation and performance analysis, will follow in future work; however, we sketch an initial outline in the following paragraph.

First, using the existing software tool \textsc{TINNiK} \cite{allman2024tinnik}, the tree-of-blobs of a network can be consistently constructed from \emph{concordance factors}: summary statistics derived from gene tree probabilities under the Network Multispecies Coalescent model. Once this tree is constructed, it suffices to consider every blob separately, and hence we focus on bloblets for the remainder of this paragraph. Specifically, the concordance factors give information about displayed quartets and they can be used to compute the pairwise NANUQ distances between the leaves of a bloblet (see again \cref{tab:canonicalNANUQdist}). Then, when input to the consistent method \textsc{Neighbor-Net} for fitting circular decomposable metrics, these distances give rise to a \emph{splits graph} \cite{bryant2004neighbor}. As we have shown in~\cref{thm:NANUQ_circular}, the NANUQ distances are circular decomposable and thus, for perfect data, the resulting splits graph depicts the splits of the displayed trees of the network and they correspond to a unique canonical form of the network (\cref{prop:canonical_equivalences}). Using concepts similar to the \emph{darts} of \cite{allman2019nanuq}, this canonical form can be derived directly from the splits graph. While we do not explicitly treat this here, as noted after \cref{thm:identifiability}, some of this is intuitive: the circular ordering of the bloblet leaves and their hybrid nodes are easily seen in most cases (see, e.g., \cref{fig:network_splitsgraph}). Following the NANUQ$^+$ framework~\cite{allman2024nanuq+}, with noisy data, a `best fit' canonical form can be chosen by a more formal criterion.

\begin{figure}[htb]
    \centering
\includegraphics[width=0.26\textwidth]{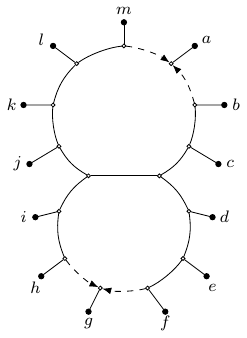}
  \qquad\qquad
\includegraphics[width=0.45\textwidth]{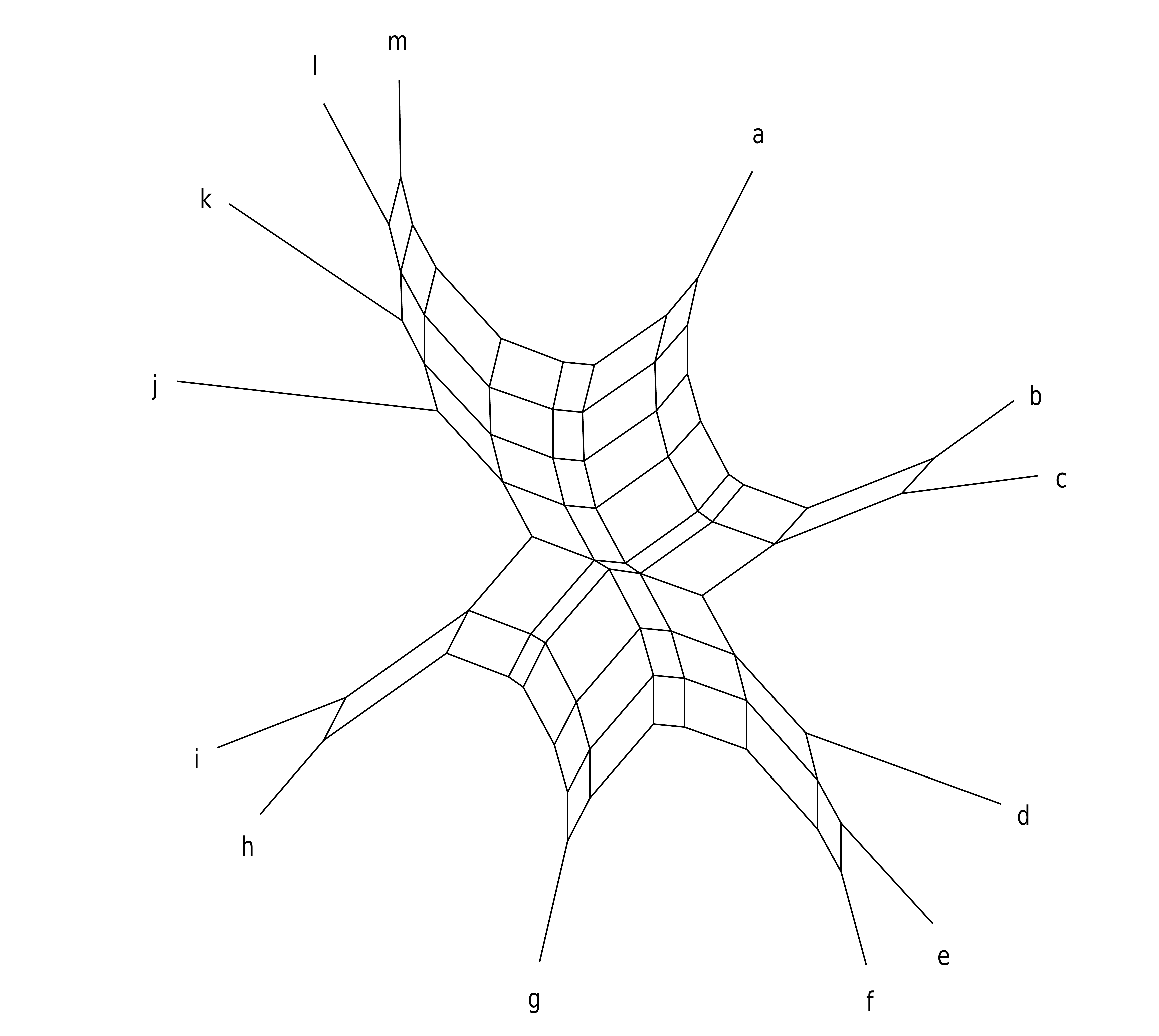}
    \caption{\emph{Left:} An outer-labeled planar, galled, level-2, semi-directed bloblet network~$N$ on leaf set $X = \{a, \ldots, m \}$. \emph{Right:} The splits graph of the pairwise NANUQ distances~$d_N$ of the network~$N$, obtained with \textsc{Neighbor-Net} \cite{bryant2004neighbor}.}
    \label{fig:network_splitsgraph}
\end{figure}

We end by stating a few open problems related to the theory developed in this work. First, we conjecture that one of our main results (\cref{thm:NANUQ_circular}) can be extended to more general classes of networks. Specifically, supported by exploratory simulations, we conjecture that this theorem extends to all level-2 outer-labeled planar galled networks (not just bloblets within this class), as is the case for level-1 networks \cite{allman2019nanuq}. Although such a result would not give new identifiability results from biological data (since the tree-of-blobs is already identifiable from quartets \cite{allman2023tree}), it could be useful for inference. In particular, it would allow for obtaining the tree-of-blobs of a network in the class $\NN_2$ directly from the NANUQ distances, instead of relying on \textsc{TINNiK} \cite{allman2024tinnik}. 
In addition, we conjecture that the NANUQ metric of outer-labeled planar, galled, level-3, bloblet networks is also circular decomposable. Proving this will either require an even more extensive case analysis than in \cref{lemma:NANUQ_weights_nonzero}, or a different technique. A final related open problem is whether some of our results can be extended to a parametric family of distances, as introduced in \cite{allman2024nanuq+}. Such a result might be used to make an inference algorithm more robust, or for developing heuristics for choosing best-fit blob structures, as in the level-1 case.

\section*{Declarations}

\paragraph{Funding} This paper is based upon work supported by the National Science Foundation (NSF) under grant~DMS-1929284 while all authors were in residence at the Institute for Computational and Experimental Research in Mathematics in Providence, RI, during the semester program on
``Theory, Methods, and Applications of Quantitative Phylogenomics''. NH and LvI were partially supported by the Dutch Research Council (NWO) grant~OCENW.M.21.306. 
ESA and JAR were partially supported by the National Science Foundation (NSF) grant DMS-205176, and HB by NSF grant DMS-2331660.

\paragraph{Data availability statement}
No data are associated with this article.

\printbibliography[title = References]

\end{document}